\documentclass[11pt]{article}
\usepackage[utf8]{inputenc}
\usepackage{lineno}

\usepackage{color}
\usepackage{graphicx}
\usepackage{amsmath}
\usepackage{amssymb}
\usepackage{float}
\usepackage{authblk}
\usepackage{geometry}
\usepackage{indentfirst}
\usepackage[bf,footnotesize]{caption}
\usepackage[colorlinks=true,allcolors=blue]{hyperref}
\usepackage[english]{babel}
\usepackage{amsmath}
\usepackage{amssymb}
\usepackage{amsthm}
\usepackage{bm}
\usepackage{booktabs}
\usepackage{changepage}
\usepackage{color}
\usepackage{eucal}
\usepackage{fullpage}
\usepackage{graphicx}
\usepackage{hhline}
\usepackage{lineno}
\usepackage{makecell}
\usepackage[sc]{mathpazo}
\usepackage{mathrsfs}
\usepackage{mathtools}
\usepackage{multirow}
\usepackage[numbers,sort&compress,square]{natbib}
\usepackage{setspace}
\usepackage{stmaryrd}
\usepackage[noend]{algpseudocode}
\usepackage{algorithmicx,algorithm}

\newcommand*\patchAmsMathEnvironmentForLineno[1]{%
	\expandafter\let\csname old#1\expandafter\endcsname\csname #1\endcsname
	\expandafter\let\csname oldend#1\expandafter\endcsname\csname end#1\endcsname
	\renewenvironment{#1}%
	{\linenomath\csname old#1\endcsname}%
	{\csname oldend#1\endcsname\endlinenomath}}%
	\newcommand*\patchBothAmsMathEnvironmentsForLineno[1]{%
	\patchAmsMathEnvironmentForLineno{#1}%
	\patchAmsMathEnvironmentForLineno{#1*}}%
	\AtBeginDocument{%
	\patchBothAmsMathEnvironmentsForLineno{equation}%
	\patchBothAmsMathEnvironmentsForLineno{align}%
	\patchBothAmsMathEnvironmentsForLineno{flalign}%
	\patchBothAmsMathEnvironmentsForLineno{alignat}%
	\patchBothAmsMathEnvironmentsForLineno{gather}%
	\patchBothAmsMathEnvironmentsForLineno{multline}%
}

\definecolor{qiblue}{rgb}{0, 0, 1}
\definecolor{sgreen}{RGB}{0, 106, 20}

\begin{document}

\newcommand\x{\textbf{x}}
\newcommand\w{\textbf{w}}
\renewcommand\ss{\{0, 1\}^{N}}
\newcommand\C{\textbf{C}}
\newcommand\D{\textbf{D}}
\renewcommand\S{\mathcal{S}}
\newcommand\rep[1]{p_{(R,\alpha)}(#1)}
\newcommand\td{\text{d}}
\newcommand\G{\mathcal{G}}
\newcommand\V{\mathcal{V}}
\newcommand\E{\mathcal{E}}
\newcommand\prob{\mathbb{P}}
\newcommand\Z{\mathbb{Z}}

\newcommand\N{\mathcal{N}}
\renewcommand\d{\text{d}}
\renewcommand\a{\textbf{a}}
\renewcommand\b{\textbf{b}}
\renewcommand\c{\textbf{c}}
\renewcommand\i{\textbf{i}}
\newtheorem{definition}{Definition}
\newtheorem{lemma}{Lemma}
\newtheorem{theorem}{Theorem}
\newtheorem{proposition}{Proposition}
\newtheorem{corollary}{Corollary}

\bibliographystyle{unsrt}
\title{The emergence of burstiness in temporal networks}

\author
{Anzhi Sheng$^{1,2}$, Qi Su$^{2,3,4}$, Aming Li$^{1,5\ast}$, Long Wang$^{1,5\ast}$, Joshua B. Plotkin$^{2,3\ast}$\\
\normalsize{$^{1}$Center for Systems and Control, College of Engineering,}\\
\normalsize{Peking University, Beijing, 100871, China}\\
\normalsize{$^{2}$Department of Biology,}\\
\normalsize{University of Pennsylvania, Philadelphia, PA 19104, USA}\\
\normalsize{$^{3}$Center for Mathematical Biology,}\\
\normalsize{University of Pennsylvania, Philadelphia, PA 19014, USA}\\
\normalsize{$^{4}$School of Mathematics and Statistics,}\\
\normalsize{University of St Andrews, St Andrews, KY16 9SS, United Kingdom}\\
\normalsize{$^{5}$Center for Multi-Agent Research, Institute for Artificial Intelligence,}\\
\normalsize{Peking University, Beijing, 100871, China}\\
\normalsize{$^\ast$Corresponding authors. Email: amingli@pku.edu.cn, \\ longwang@pku.edu.cn, jplotkin@sas.upenn.edu}\\
}
\date{ }
\maketitle

\begin{abstract}
Human social interactions tend to vary in intensity over time, whether they are in person or online. Variable rates of interaction in structured populations can be described by networks with the time-varying activity of links and nodes. One of the key statistics to summarize temporal patterns is the inter-event time (IET), namely the duration between successive pairwise interactions. Empirical studies have found IET distributions that are heavy-tailed (or ``bursty"), for temporally varying interaction, both physical and digital. But it is difficult to construct theoretical models of time-varying activity on a network that reproduces the burstiness seen in empirical data. Here we develop a spanning-tree method to construct temporal networks and activity patterns with bursty behavior. Our method ensures a desired target IET distribution of single nodes/links, provided the distribution fulfills a consistency condition, regardless of whether the underlying topology is static or time-varying. We show that this model can reproduce burstiness found in empirical datasets, and so it may serve as a basis for studying dynamic processes in real-world bursty interactions.
\end{abstract}

\section{Introduction}
Temporal networks have long been recognized as a powerful tool to model complex systems with time-varying interactions \cite{holme2012, holme2015, masuda2016guide}.
A large body of literature concentrates on analyzing the activation dynamics of nodes and links in such networks.
The inter-event time (the waiting time between two consecutive interaction events, IET) is a canonical measure of temporal patterns, and it is known to have profound effects on individual behavior \cite{karsai2018bursty, karsai2012universal, gernat2018automated, choi2020individual} and dynamical processes occurring on networks \cite{karsai2011small, jo2014analytically, lambiotte2013burstiness, mancastroppa2019burstiness, li2020}.
A variety of empirical datasets, such as email and mobile communications \cite{onnela2007structure, malmgren2008poissonian, malmgren2009universality}, epidemic transmission \cite{vazquez2006polynomial,vazquez2007impact, unicomb2021dynamics}, and human mobility \cite{song2010modelling, schlapfer2021universal}, exhibit non-Poisson-like  activity patterns, known as burstiness \cite{karsai2018bursty, barabasi2005origin, vazquez2006modeling}.
These IET patterns are characterized by periods of frequent activation interleaved with long periods of silence.  
Empirical networks often exhibit burstiness in both the activity of individuals (nodes) as well as interactions (edges) \cite{karsai2012correlated, saramaki2015seconds, genois2018can}.

It has proven difficult to construct synthetic temporal networks whose properties are similar to the bursty behavior seen in empirical temporal networks \cite{goh2008burstiness}.
Previous approaches can be divided into two categories: structure-based modeling, and contact-based modeling.
The former approach applies dynamical processes on static underlying topologies, such as random walks \cite{starnini2012random, barrat2013modeling}, link dynamics \cite{holme2013epidemiologically}, and inhomogeneous Poisson processes \cite{hiraoka2020modeling}.
For example, Barrat et al used random itineraries on weighted underlying networks to generate time-extended structures with bursty behavior \cite{barrat2013modeling}. 
The latter approach uses a stream of contacts generated by certain realistic mechanisms, such as social appeal \cite{starnini2013modeling}, individual resource \cite{aoki2016temporal}, and memory \cite{vestergaard2014memory}.
For example, Perra et al propose the activity-driven model \cite{perra2012activity}, in which each node, isolated in the beginning, becomes active with a probability proportional to its own activity potential and forms links with other nodes.
Despite a large body of studies in constructing temporal networks, they usually fail to reproduce the level of burstiness as empirical datasets, and they lack mathematical guarantees for the behavior of the synthetic network.

In this study, we provide an analytical framework to systematically construct temporal networks on both static and time-varying underlying topologies.
Our construction algorithm can reproduce the burstiness of both nodes and edges in four empirical datasets, including social interactions in rural Malawi, colleague relationships in an office building over two years, and friendship relations in a high school. 
The assumptions of our model can also be tested in the empirical datasets.
Our construction thus serves as an efficient method to generate realistic temporal networks that can then be used to investigate dynamical processes (such as evolutionary dynamics, social contagion, or epidemics) on temporal networks.

Our approach to constructing temporal networks uses a spanning-tree method.
Spanning trees are widely recognized as a significant family of sparse sub-graphs since they tend to govern dynamical processes on full graphs \cite{newman1994spin, barabasi1996invasion, kossinets2008structure, garlaschelli2003universal}.
For example, in social networks, the backbone of an aggregated communication topology is often constructed as the union of shortest-path spanning trees, on which information flows fastest \cite{kossinets2008structure}. As a result, a large portion of directed edges are bypassed by faster indirect routes in the tree. In studies of food-webs  \cite{garlaschelli2003universal}, spanning trees are defined as the flows from the environment to every species. 
The links in or out of the tree are denoted as ‘strong’ or ‘weak’ links, related to delivery efficiency or system robustness and stability.
 
\section{Spanning-tree method}
We introduce a spanning-tree method for constructing temporal networks on any underlying topology, which restricts the interaction pattern between pairs of individuals. 
The activity of every single node and link is a binary-state process, switching between active and inactive.
We use inter-event time (IET) distribution, which measures the time intervals between consecutive activations, to quantify the activity patterns of nodes and edges.

Our method allows for both static and time-varying underlying topologies.
Here we first consider static underlying topologies, in which the activation dynamics on a given topology is much faster than the evolution of the underlying topologies, such as the time-varying traffic flow on a relatively stable road network.
We consider three classes of topologies in order of increasing complexity: two-node topologies, tree topologies, and finally arbitrary structured topologies. 

\subsection{Two-node systems}
We begin with a basic unit of a networked system -- a two-node system, which is composed of nodes $x$ and $y$ and a link $z$ between them (see Fig.~1A). The nodes and edges are either active or inactive at each discrete time step.
We assume that the state updating of $x$ follows a renewal process $\{X_n \}_{n\ge 0}$ and assigns $x$ a probability mass function $F(\Delta t)$ as its target IET distribution (see Supplementary Material section 1). The random variable $X_n$ equals $1$ if $x$ is active at time $n$, otherwise $X_n = 0$. Likewise for node $y$ and edge $z$, which have respective target IET distributions   $G(\Delta t)$ and $ H(\Delta t)$, and respective renewal processes  $\{Y_n \}_{n\ge 0}$ and $ \{Z_n \}_{n\ge 0}$. The initial state of $x,y,z$ is active (i.e. $X_0=Y_0=Z_0=1$). The goal is to construct a two-node temporal network that satisfies the target IET distributions of $x$, $y$, and $z$. 

By definition, we say that edge $z$ is active when $x,y$ are both active (i.e. $Z_n = X_nY_n$).
Furthermore, we assume that, given the trajectory of $x$ until $n$, the probability of $x$ being active at time $n+1$ is independent of the trajectory of $y$ until $n$, that is, 
\begin{equation}
	\prob(X_{n+1},Y^{(n)}|X^{(n)}) = \prob(X_{n+1}|X^{(n)})\cdot \prob(Y^{(n)}|X^{(n)}).
\end{equation}
Given these assumptions, we can show that when all targeted distributions are identical (i.e. $F=G=H$) the system must be completely synchronous, that is, all of $x,y,z$ are either active or inactive at each time step (see Supplementary Material section 2).

\begin{figure}[p]
\centering
\includegraphics[width=1\textwidth]{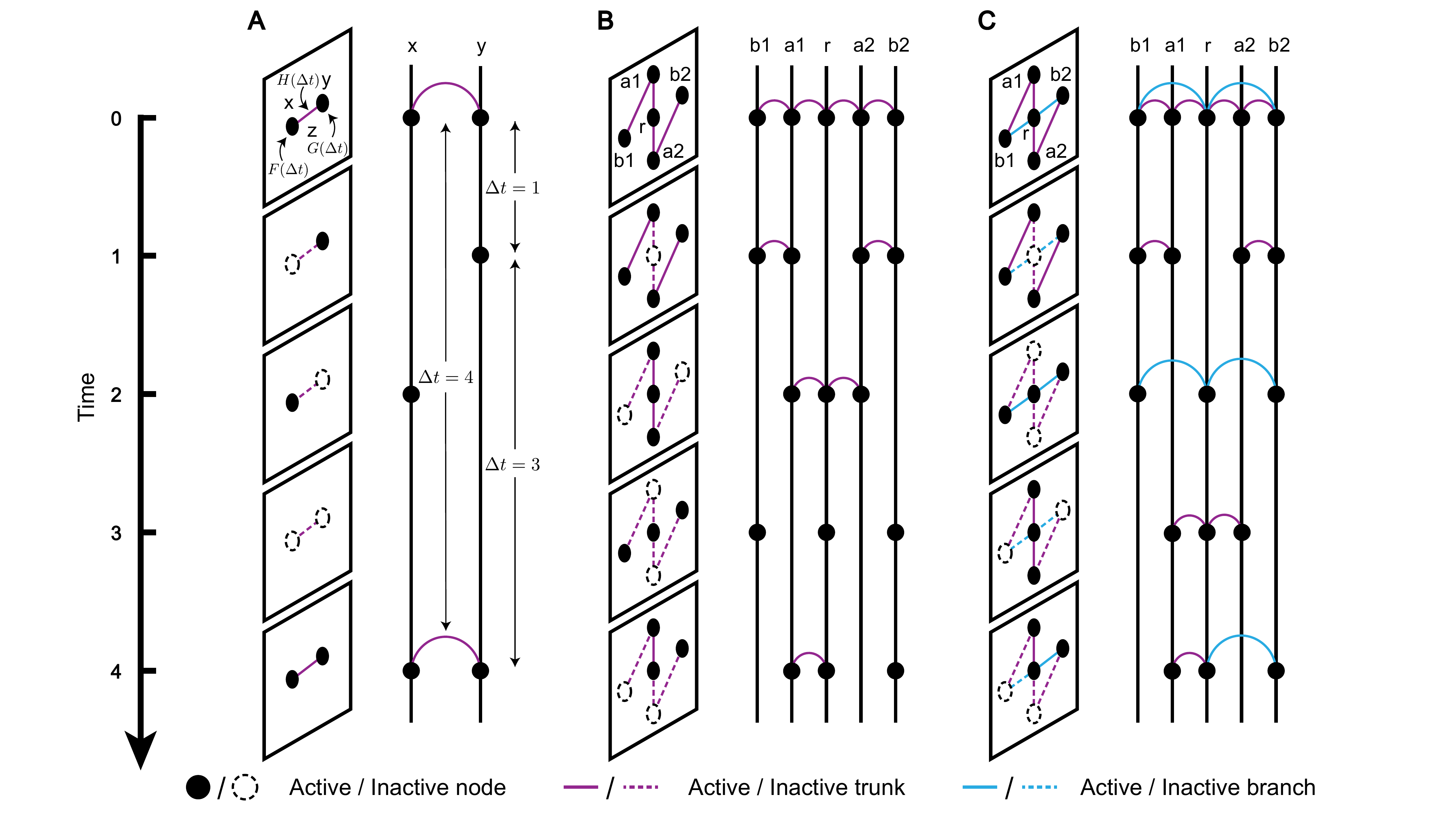}
\caption{\textbf{Schematic illustration of constructing temporal networks on different underlying topologies.}
Each node/link switches between two states, active (solid circle/line) and inactive (dashed circle/line), and all nodes and links are set to be active initially.
(\textbf{A}) The basic unit of a network system is a two-node system with nodes $x$, $y$, and a link $z$ connecting them. 
At the beginning of temporal network construction, three probability mass functions $F(\Delta t), G(\Delta t), H(\Delta t)$ are given as the expected IET distributions for $x,y,z$, where $\Delta t$ represents the time interval between two consecutive activations. Then, the activity of $x,y,z$ is driven by the renewal processes with the corresponding expected IET distributions.
We constraint that $z$ is active only if both $x$ and $y$ are active.
(\textbf{B}) An extension of two-node systems is tree systems, in which nodes are divided into two categories, a root ($r$) and leaves ($a1, a2, b1, b2$). 
The state of nodes and links is updated by sequentially executing the algorithm over each two-node system from the root ($r$-$a1$ and $r$-$a2$) to the outermost leaves ($a1$-$b1$ and $a2$-$b2$). 
(\textbf{C}) There is at least one spanning tree for any static underlying topology. 
The links in the spanning tree are called trunks (purple lines) and the links outside the spanning tree are called branches (blue lines). 
The states of nodes and trunks are updated first, and then the states of branches are established according to the state of the nodes on both sides.
}
\end{figure}

Given the state of the first $n$ times (from time $0$ to $n-1$), the conditional probability that node $x$ is active at time $n$ is given by
\begin{equation}
	\prob(X_n=1|X_{n-1}=w^{(n-1)}_{x},...,X_{0}=w^{(0)}_{x}):=p_x(\w^{(n-1)}_x,1) = \frac{F(n-m)}{\sum_{i\ge n-m} F(i)}.
\end{equation}
where $\w_x^{(n-1)}=(w^{(0)}_{x},...,w^{(n-1)}_{x})^{\text{T}} \in \{0,1\}^{n}$ records all historical states of node $x$ before time $n$, called $x$'s trajectory, and $m=\max \{ k\le n: w^{(k)}_{x}=1 \}$ represents the last activation time of $x$.
Analogous conditional probabilities apply to $y$ and $z$.

We propose an algorithm to construct two-node temporal networks, such that the  IET distributions of $x,y,z$ fit the desired targeted distributions $F, G, H$ and the desired total time duration $t_{tol}$.
At each time step $t=n+1$ ($0 \le n \le t_{tol}-1$), we calculate four probabilities with Eq.~2,
\begin{equation}
\begin{aligned}
		&p_1 = p_z(\w^{(n)}_z,1),\ p_2 = p_x(\w^{(n)}_x,1) - p_z(\w^{(n)}_z,1), \\
		&p_3 = p_y(\w^{(n)}_y,1) - p_z(\w^{(n)}_z,1),\ p_4 = 1 + p_z(\w^{(n)}_z,1) - p_x(\w^{(n)}_x,1) - p_y(\w^{(n)}_y,1).
\end{aligned}
\end{equation}
$p_1$ represents the probability that $z$ is active at $t$, and $p_2$ (respectively $p_3$ and $p_4$) represents the probability that $z$ is inactive when $x$ is active (respectively $y$ is active and both $x$ and $y$ are inactive).
Next, we determine the state of $x$ and $y$ at $t$, in order.
The probability that $x$ is active is $p_x(\w^{(n)}_x,1)$.
If $x$ is active, the probability of $y$ being active is $p_1/ p_x(\w^{(n)}_x,1)$.
Otherwise, the probability becomes $p_3/(1- p_x(\w^{(n)}_x,1))$.
Finally, we update the trajectory of $z$ by the relation $w_z^{(t)}=w_x^{(t)} w_y^{(t)}$.
The construction stops when $t=t_{tol}$.
Algorithm 1 in Supplementary Material outlines the above procedure.
It is worth noting that the activation order of $x$ and $y$ does not affect the IET distributions of $x,y,z$.

Although the algorithm is specified for an arbitrary combination of target IET distributions $(F, G, H)$, these distributions must  satisfy an implicit condition in order to guarantee the consistency of the construction. For example, if the target distributions specify that the link $z$ is activated more frequently than the nodes $x$ and $y$, then no construction is possible, because the link is active only when both nodes are active.
This would result in at least one of the probabilities $p_i$ ($i=1,2,3,4$) in Eq.~3 being less than $0$ during the construction.
We say that the combination $(F, G, H)$ of target distributions is consistent if $p_i$ ($i=1,2,3,4$) belong to $[0, 1]$ for all possible trajectories $\w_x^{(n)}$ and $\w_y^{(n)}$ with any length $n$.
When a two-node system is consistent, the algorithm is well-defined, and it provably ensures that the IET distributions of $x,y,z$ will satisfy the targets $F, G, H$, respectively (see Supplementary Material section 2).

\subsection{Tree systems}
The construction for two-node systems can be naturally extended to tree systems, which consist of a number of interconnected two-node systems (Fig.~1B).
We randomly select a node as the root $r$ and classify the remaining nodes (i.e. leaves) according to their distance from $r$. 
We assign each node and link a targeted IET distribution.
At each time step, we first determine the state of $r$, which is only related to its own trajectory.
Then, every leaf one step away from $r$ ($a1$ and $a2$ in Fig.~1B) forms a two-node system with $r$, and the states of these leaves are determined by Algorithm 1.
Next, all leaves at the first layer form two-nodes systems with their corresponding leaves at the second layer [$(a1,b1)$ and $(a2,b2)$ in Fig.~1B].
Analogously, the states of all leaves are updated within a two-node system in the order of their layers. 
Algorithm 2 in Supplementary Material summarizes this procedure.

This procedure requires an additional assumption of conditional independence -- that for a pair of two-node systems sharing a node, if the state of the common node is given then the activity of the other two nodes are independent. As shown in Fig.~1B, we present two examples, ($r$, $a1$, $a2$) and ($r$, $a1$, $b1$), fulfilling the condition. The former indicates that the activity of nodes on the same layer ($a1, a2$) is independent when the state of the node in the upper layer ($r$) is settled.
The latter indicates that the activity of a node ($b1$) is not affected by the node that is more than one layer away ($r$), given the state of the middle node ($a1$).

It is straightforward to show that a tree system is consistent if all two-node systems in the tree are consistent. Similarly, when a tree system is consistent, the algorithmic IET distribution of every single node/link fulfills its corresponding target distribution (see Supplementary Material section 2). 

\subsection{Spanning-tree based construction}
For any static (but arbitrary) underlying topology, we can always find a spanning tree  (Fig.~1C). 
A link is called a trunk if it is in the spanning tree, and it is called a branch otherwise. 
In the construction algorithm, we first select a spanning tree and choose a target IET distribution for each node and trunk.
Next, we execute Algorithm 2 on the spanning tree, so that all nodes and trunks can update their states to achieve the targeted IETs.
The state of each branch is then active when the nodes on both ends are active.
As a result, the activity of the system is determined by its spanning tree.
Algorithm 3 in Supplementary Material summarizes this procedure.

We say that the system is consistent when its spanning tree system is consistent.
In particular, when all two-node systems in the spanning tree are synchronous, the trajectories of all nodes and links (including trunks and branches) are identical, and the whole system is synchronous.

\section{Applications}
\subsection{Synthetic temporal networks}
To test our algorithm, we construct temporal networks with one of two activity patterns -- bursty activity patterns and Poisson activity patterns. Bursty patterns arise when there is a simultaneous bursty activity in node and links activity, and Poisson pattern arises in settings such as bank queuing systems \cite{ross1996stochastic} and spreading dynamics \cite{barthelemy2004velocity}.
In particular, we choose target IET distributions of nodes and trunks that are power-law distributions for bursty activity patterns, given by
\begin{equation}
	p(\Delta t; \alpha) \sim \Delta t^{-\alpha} \quad (\alpha > 1).
\end{equation}
And we choose discrete exponential distributions for Poisson-like activity patterns, given by
\begin{equation}
	p(\Delta t; \alpha) \sim \int_{\Delta t - 1/2}^{\Delta t + 1/2}\alpha e^{-\alpha x} \text{d}x \quad (\alpha > 0),
\end{equation}
where $\Delta t$ represents the inter-event time and $\alpha$ the exponent.
The respective survival functions are given as
\begin{equation}
	\prob(T > \Delta t) \sim \Delta t ^{-\alpha+1}
\end{equation}
with exponent $\alpha-1$ and 
\begin{equation}
	\prob(T > \Delta t) \sim e^{-\alpha \Delta t}
\end{equation}
with exponent $\alpha$.

In the main text, we focus on a simple case when all nodes (trunks) have a common exponent $\alpha_{pmf}$ ($\beta_{pmf}$), and we use the aggregated IET distribution \cite{barabasi2005origin, vazquez2006modeling, karsai2012correlated}, that counts the IETs of all nodes or links, to quantify the intensity of activity.
In Supplementary Material, we also consider the IET distributions of every single node and link (Fig.~S1) and we investigate the relationship with the aggregated IET distributions. We also explore a more general case in which the exponents of nodes and trunks are sampled independently from a distribution (Fig.~S2). 

\begin{figure}[h]
\centering
\includegraphics[width=1\textwidth]{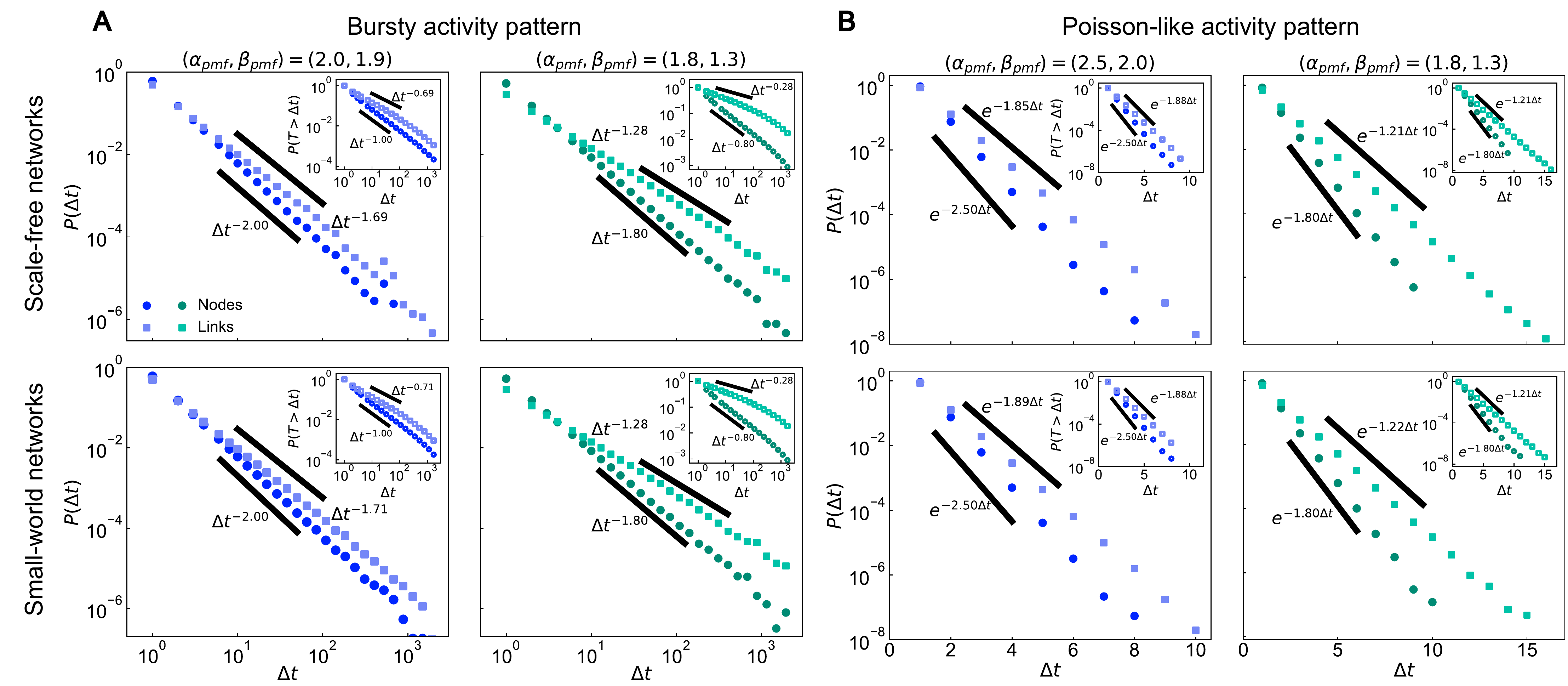}
\caption{\textbf{Aggregated IET distributions of nodes and links on different underlying topologies.}
We consider the construction of temporal networks on two classes of static underlying topologies, Barabási-Albert scale-free networks (first row) and Watts-Strogatz small-world networks (second row).
The targeted exponent of every single node (trunk) is identical, denoted as $\alpha_{pmf}$ ($\beta_{pmf}$).
We select a pair of high exponents (blue dots) and a pair of low exponents (green dots) for the bursty activity pattern (\textbf{A}) and the Poisson-like activity pattern (\textbf{B}), given as $(\alpha_{pmf}, \beta_{pmf}) = (2.20, 2.10), (1.80, 1.30)$ and $(\alpha_{pmf}, \beta_{pmf}) = (2.50, 2.00), (1.80, 1.30)$, respectively.
The algorithmic distributions of nodes (circles) and links (squares) are well predicted by power-law distributions in (\textbf{A}) and by exponential distributions in (\textbf{B}). 
The thick black lines with fitted exponents are plotted for reference.
The distributions are the average over $50$ independent trials.
Parameter settings: network size $N=10^3$, average degree $k=6$, and network length $t_{tol}=10^4$ in (\textbf{A}) and $t_{tol}=10^3$ in (\textbf{B}).
}
\end{figure}

We begin our analysis by constructing temporal networks with bursty activity patterns.
We derive a necessary and sufficient condition for system consistency, which applies to any network length $t_{tol}$ (see Materials and methods Eq.~13 and Supplementary Material for details). 
As examples, we consider two pairs of exponent setups, $(\alpha_{pmf}, \beta_{pmf}) = (2.20, 2.10), (1.80, 1.30)$, and we execute Algorithm 3 over two classes of static underlying topologies, Barabási-Albert scale-free networks \cite{barabasi1999emergence} and Watts-Strogatz small-world networks \cite{watts1998collective}.
Figure 2A shows the aggregated IET distributions of nodes and links.
For both underlying topologies and both parameter setups, the probability mass functions $\prob(\Delta t)$ and the corresponding survival functions $\prob(T > \Delta t)$ are well fit by power-law distributions, showing simultaneous burstiness in nodes and links.
The best-fit exponent for nodes matches the exponent of the target IET distribution, and the exponent for links is slightly lower than the target exponent, due to the impact of branches (edges outside the spanning tree, where the algorithm is guaranteed to work).

Next, we construct temporal networks with Poisson-like activity patterns. If the target distributions for nodes and trunks are Poisson distributions, then we can prove that the system is never consistent (see Supplementary Material section 2). However, it is possible to construct consistent systems  when the target IET follows  discrete exponential distributions. In this case, we derive a necessary and sufficient condition for system consistency -- namely, that the difference of $\alpha_{pmf}$ and $\beta_{pmf}$ lies in $[0, \ln2]$, meaning that the activity of nodes must be more frequent than links but not too frequent (see Materials and methods and Supplementary Material for details).
Figure 2B shows the aggregated IET distributions along with the exponents of the target distributions $(\alpha_{pmf}, \beta_{pmf}) = (2.50, 2.00)$ and $ (1.80, 1.30)$. All the distributions follow the expected exponential decay.


Comparing the results for the two topologies, we find that the exponents produced by the algorithmic construction are sensitive to the choice of target IET distribution, but robust to the choice of network topology.
To examine this observation more generally, we investigated a wide range of random regular underlying topologies with different average degrees, ranging from $5$ to the well-mixed case (i.e., each node linked to all other nodes, see figs.~S3 and S4). We find that we can robustly match target IET distributions across all these topologies: the relative deviation between the largest and smallest exponent in allegorically constructed networks is within $4\%$.

We can understand why the algorithmic construction for producing a desired IET distribution is robust to topology by analyzing the activity of branches.
In particular, we prove that the IET distribution of every single branch is approximately a power-law distribution (respectively a discrete exponential distribution) in bursty activity patterns (respectively Poisson-like activity patterns) with uniform upper and lower bounds related only to the target distributions (see Supplementary Material section 3).
And we can confirm that the aggregated IET distributions are also robust to the selection of spanning trees (fig.~S5).

\begin{figure}[t]
\centering
\includegraphics[width=1 \textwidth]{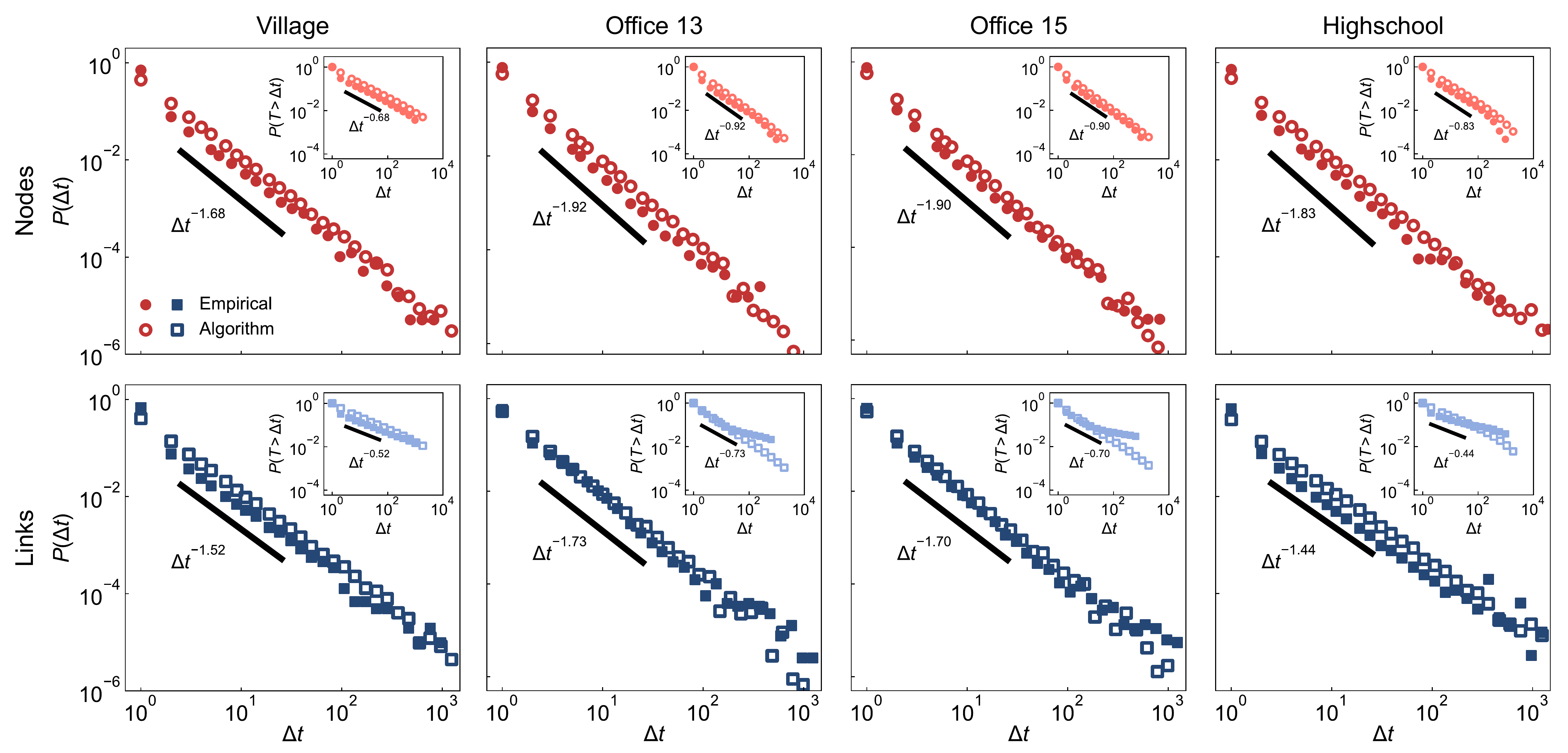}
\caption{\textbf{Burstiness in empirical datasets.}
We analyze four empirical temporal networks of social interactions within and across households among $84$ individuals in a village \cite{ozella2021using}; colleague relationships among $95$ and $219$ employees in an office building in two different years (2013 and 2015) \cite{genois2015data, genois2018can}; and friendship and educational relationships among $327$ students in a high school in Marseilles \cite{mastrandrea2015contact}. 
The length of these temporal networks from left to right are 43436, 20129, 21536, and 7375. 
For each empirical temporal network, we count the aggregated ICT distributions of nodes (solid circles) and links (solid squares), which both present bursty behaviors and are well-predicted by power-law distributions.
The thick black lines are power-law distributions with the fitted exponents.
We take the fitted distributions as targets and obtain the respective algorithmic distributions of nodes (hollow circles) and links (hollow squares).
The algorithmic distributions present the same level of burstiness as the corresponding empirical distributions.
}
\end{figure}  

\subsection{Empirical temporal networks}
We tested the ability of our algorithm to reproduce the burstiness of activity patterns observed in four empirical datasets, collected by the SocioPatterns collaboration \cite{web}.
These four datasets record pairs of face-to-face interactions from different social contexts, ranging from a village in rural Malawi, to an office building and a high school in France.
Each dataset is comprised of contact events with timestamps, represented by triplets $(t,i,j)$ -- indicating the occurrence of an interaction between individual $i$ and individual $j$ at time $t$.

It is worth noting that the empirical data record only the communication moments, so that only the active nodes with at least one active neighbor can be detected. In other words, the empirical data are observations of the inter-communication times (ICTs), rather than the IETs of nodes. 
Nevertheless, we demonstrate that the ICT distribution of single nodes converges exponentially to the IET distribution as the number of neighbors on the underlying topology increases (see Supplementary Material section 4).
Since empirical datasets often originate from highly connected populations, the ICT distributions approximate the statistical properties of the corresponding IET distributions.

Before applying our algorithm, we first test whether the assumptions underlying the algorithmic construction are consistent with the empirical datasets.
The assumption that the activity of nodes and links is a renewal process is reasonable, compared to the empirical data (fig.~S6). However,  the strong form of the conditional independence (Eq.~1) assumed by our construction is rejected for the empirical data (fig.~S7). Nonetheless, the empirical data satisfy a weaker form of conditional independence (fig.~S8, see Supplementary Material section 5 for details).

After pre-processing the datasets (see Materials and methods), we obtain four empirical temporal networks with population sizes ranging from $N=84$ to $N=327$ and length from $t_{tol}=7,375$ to $t_{tol}=43,436$ time-steps.
For all of these temporal networks, the empirical ICT distributions of nodes and links both exhibit heavy tails, with different decay rates, showing simultaneous burstiness in activity.
We fit these empirical ICTs with power-law distribution (Eq.~4) by maximum likelihood estimation \cite{clauset2009power, alstott2014powerlaw},  and we use the fitted distributions as targets for constructing synthetic temporal networks.
Figure 3 shows the comparison between the empirical and algorithmically constructed ICT distributions of nodes and links. 
Our algorithm  successfully replicates the qualitative patterns of burstiness observed in empirical datasets.
Figure S9 shows the comparison between the IET and ICT distributions of nodes.
Since the average degree of these empirical underlying topologies is large, the IET distribution collapses onto the ICT distribution.

\subsection{Combination with network evolution}

\begin{figure}[t]
\centering
\includegraphics[width=1 \linewidth]{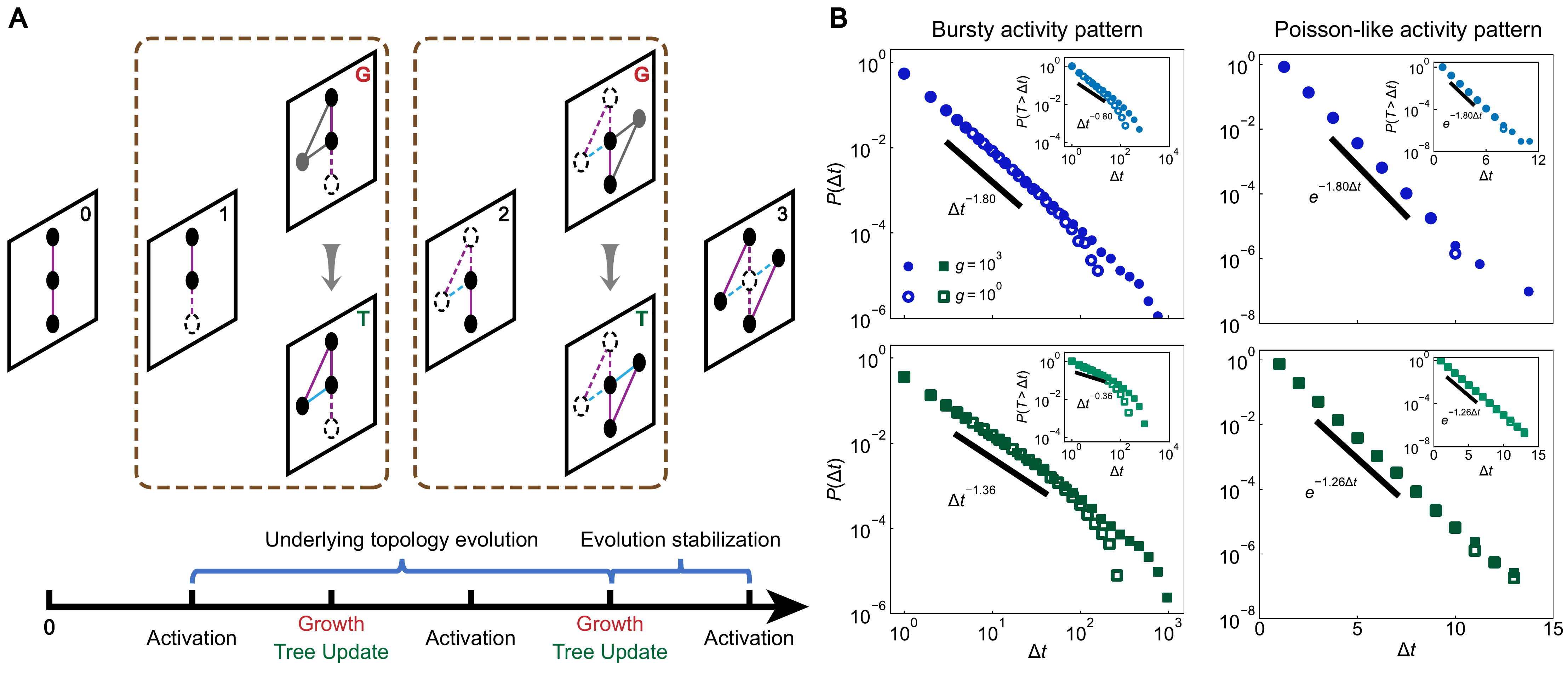}
\caption{\textbf{Construction on time-varying underlying topologies.} 
(\textbf{A}) We consider the temporal network construction with a time-varying topology modeled by the Barabási-Albert model.
There are $m_0$ nodes in the initial snapshot. 
When the underlying topology is evolving, a node with $m$ link(s) enters the network system (snapshot G), and the spanning tree is updated accordingly (snapshot T).
When the evolution is stable, the construction process degenerates to that for the static underlying topology.
(\textbf{B})  
We consider the impact of the duration time $g$ after network evolution stabilization.
Our model produces expected bursty and Poisson-like activity patterns on such a time-varying underlying topology.
Furthermore, the activity patterns formed during network evolution (hollow dots) are maintained after evolution stabilization (solid dots).
Parameter settings: $\alpha_{pmf}=1.8$, $\beta_{pmf}=1.5$, $m_0=3$, $m=3$, and final network size $N=500$.}
\end{figure}

Although some underlying topologies are static, a variety of real-world systems also exhibit topology changes over timescales that are comparable to the activation dynamics on the network.
For example, in online social networks, new users can enter the network and engage in new interactions with existing users; or users can switch between online and offline states.
As a consequence, the temporal changes in activity originate not only from the states of existing nodes and links, but also from the addition and subtraction of nodes and links in the underlying structure. 

With this as a backdrop, we extend our algorithm from static to dynamic underlying topologies, starting first with networks that grow in size. We introduce a new model that combines our algorithm for constructing temporal networks with the Barabási-Albert model \cite{barabasi1999emergence}, which we call the temporal Barabási-Albert model.
The construction process is as follows.
An underlying topology is initialized with $m_0$ nodes, and the spanning tree is selected randomly.
At each time step 
the activity state of the existing network updates once with Algorithm 3, then one adds a new node with $m\le m_0$ links connected to $m$ different existing nodes following the preferential attachment rule \cite{barabasi1999emergence}.
All the newly added elements are set to be active, and the states of the older pre-existing elements are updated accordingly.
The spanning tree is then updated by adding the new node and a link randomly selected from $m$ new links.
At some time point, the underlying topology stops growing, and the construction process continues with $g\ge 0$ more steps on the final state of the underlying topology.
Figure 4A shows a schematic illustration of this above procedure. 
Figure 4B shows numerical simulations of the temporal Barabási-Albert model in bursty and Poisson-like activity patterns.
We find that the IET distributions of nodes and links are robust to the duration time $g$, which means that the activity pattern is established during the evolution of the underlying topology, and it is then preserved after the topology is fixed.

In addition to network growth, nodes and links may also be removed, due to aging effects and other recessionary impacts.
To model these various kinds of network evolution, we introduce a more general procedure for constructing temporal networks on time-varying underlying topologies (see Supplementary Material Algorithm 4).
According to this perspective, a contact-based method such as the activity-driven model can be viewed as a structure-based method with specific underlying topology evolution.  

\section{Discussion}
Simple models that neglect temporal variation in individuals' behaviors do not suffice to describe the dynamics of many real-life complex systems.
A large and growing body of studies supports that state switching of individuals and interactions plays a significant role in diverse dynamical processes, such as face-to-face communication \cite{genois2018can}, evolutionary dynamics \cite{li2020, sheng2021}, and network control \cite{li2017fundamental}.  
We have proposed an analytical framework and corresponding spanning-tree method to construct temporal networks with specific activity patterns, including bursty and Poisson-like activity patterns. The algorithm is able to reproduce the simultaneous burstiness observed in empirical datasets from diverse social contexts. 

The central ingredient in our construction algorithm -- the spanning tree -- has been widely recognized as a significant feature in both theoretical  \cite{wang2009divide, xiao2006state, xiao2008asynchronous} and real-world applications of network science \cite{parsakhoo2016determining}.
For example, in path-finding algorithms such as Dijkstra's algorithm \cite{dijkstra1959note} and the A* search algorithm \cite{hart1968formal}, the shortest paths from a given source node to all other nodes, together with the source node, form a shortest-path tree.
These algorithms are widely used in mobile robot covering problems \cite{gabriely2001spanning}, for tracking the establishment of oil pipelines \cite{baeza2017comparison}, and for vehicle routing \cite{imran2009variable}. 
The tree structure is the backbone of these networked systems.
Other examples include telecommunication networks, including the Internet, where the Spanning Tree Protocol \cite{braem2006wireless} and Augmented tree-based routing \cite{caleffi2007augmented} are used to avoid routing loops, to solve the scalability problem, and to gain resilience against node failure and link instability.
And in social networks, spanning tree-based algorithms have been proven effective in detecting communities, one of the most widely studied issues in network science \cite{basuchowdhuri2014detection}. 
And so our spanning tree-based method for generating specific activity patterns might have implications in several areas of application,  which remain to be investigated.

Our analysis of activity patterns that are robust to underlying topology indicates that activity patterns of a network system is strongly determined by the dynamics of its spanning tree.
Branches from the spanning tree cause only small perturbations to the activity patterns, regardless of their number and location, 
showing a similar function to the 'weak' links in food webs \cite{garlaschelli2003universal}. 
As a result of this robustness, we can construct a temporal network with a desired consistent activity pattern, even if the underlying topology is not precisely known, or even changing in time.

Another straightforward way to measure the burstiness and memory of temporal networks is to calculate the burstiness parameter \cite{goh2008burstiness} and the auto-correlation function (see Supplementary Material section 7).
A larger burstiness parameter means a higher level of burstiness, and a lower absolute value of the auto-correlation function means a weaker dependence on memory.
Our results show that our synthetic bursty activity patterns have strong auto-correlation and a positive and high burstiness parameter, while the Poisson-like activity patterns are memoryless and have a negative burstiness parameter (tables S1 and S2, fig.~S10).

In addition to burstiness, real-world datasets often exhibit robust behaviors across different time scales \cite{barrat2013modeling, gautreau2009microdynamics}.
In this study, we concentrate on the node strength distribution of aggregated networks (fig.~S11) and find that the distributions with different aggregation times collapse onto a common baseline distribution in each empirical dataset (fig.~S12).
Under our model of constructing temporal networks we can understand that this robustness comes from the long-term regular activation of nodes and links (see Supplementary Material section 8 for detailed derivations).
The model also reproduces the empirical observation for stationary distributions  (fig.~S13).

The past ten years have shown increasing interest in understanding the effects of group interactions and higher-order interactions \cite{mayfield2017higher, paul2020role, battiston2021physics, alvarez2021evolutionary, lambiotte2019networks}, meaning behavioral activities that are not limited to just pairs of individuals.
A recent study has shown that higher-order interactions in empirical datasets display similar bursty behaviors to pairwise interactions \cite{cencetti2021temporal}.
And so a natural extension of this work is to study the activity of higher-order interactions in temporal networks. Our approach may provide a method to decompose networks into several elementary components, analogs of spanning trees in the context of hyper-graphs, which remains a direction for future research on the temporal dynamics in groups of interacting agents.

\section{Materials and methods}
\subsection{Mathematical formalization}
Here we provide a mathematical model of the two-node temporal network construction, which is a stochastic process $\{S_n \}_{n\ge 0}$ coupling the activity of every unit. Complete mathematical details about the existence of $\{S_n \}_{n\ge 0}$ and the modeling of other systems are provided in the Supplementary Material.

We follow the notation in the Two-node systems section.
According to the constraint $Z_n=X_n Y_n$, we construct a stochastic process $\{S_m \}_{m\ge 0}$, which for arbitrary sets of $t_1,...,t_k \in \mathbb{N}, k \in \Z^+$ satisfies
\begin{equation}
\begin{aligned}
	&\mu_{S_{2t_1},...,S_{2t_k}} = \mu_{X_{t_1},...,X_{t_k}}, \\
	&\mu_{S_{2t_1+1},...,S_{2t_k+1}} = \mu_{Y_{t_1},...,Y_{t_k}}, \\
	&\mu_{(S_{2t_1}\cdot S_{2t_1+1}),...,(S_{2t_k}\cdot S_{2t_k+1})} = \mu_{Z_{t_1},...,Z_{t_k}},
\end{aligned}
\end{equation}
and $S_0=S_1=1$. Here $\mu_{X_{t_1},...,X_{t_k}}$ represents the finite dimensional distribution of $\{S_m \}_{m\ge 0}$ at the time slice $(t_1,...,t_k)$.
$\{S_m \}_{m\ge 0}$ can be viewed as composing of the following sequence
\begin{equation}
	(S_0,S_1,...,S_{2n},S_{2n+1},...)=(X_0,Y_0,...,X_n,Y_n,...),
\end{equation}
and $Z_n = S_{2n} S_{2n+1}$. 
$\{S_m \}_{m\ge 0}$ follows the activation order in the main text (i.e., the state of $x$ is determined first).
If we exchange the order, the corresponding indexes in Eq.~8 and the sequence of $X_n$ and $Y_n$ in Eq.~9 are also swapped.

\subsection{System consistency}
The consistency of the two-node system is equivalent to the existence of $\{S_n \}_{n\ge 0}$.
If $\{S_n \}_{n\ge 0}$ is well-defined, the algorithmically produced IET distributions of $x,y,z$ will satisfy the target distributions $F,G,H$.

We derive the equivalence of the consistency condition for bursty and Poisson-like activity patterns in a two-node system.
Let $p_x^{(n)} = p_x(\textbf{0}^{(n-1)},1)$
denote the conditional probability that $x$ is active for the first time at $n$, where 
\begin{equation}
	\textbf{0}^{(m-1)} = (1,\underbrace{0,...,0}_{m-1})^{\text{T}}
\end{equation}
is a trajectory with length $m$ and only one active state occurring at the initial time.
For a power-law distribution with exponent $\alpha$, 
\begin{equation}
	p_x^{(n)} \approx 1 - (\frac{n+\frac{1}{2}}{n-\frac{1}{2}})^{-\alpha+1},
\end{equation}
and for a discrete exponential distribution with exponent $\alpha$,
\begin{equation}
	p_x^{(n)} = 1 - e^{-\alpha}.
\end{equation}
From the definition of the distribution consistency, the equivalence is given as
\begin{equation}
	\alpha_{node} \ge \alpha_{link},\ p_x^{(1)}  + p_y^{(2)}< 1, \ p_x^{(1)}  + p_y^{(1)} < 1 + p_z^{(1)}
\end{equation}
for bursty activity patterns and
\begin{equation}
	0\le \alpha_{node}-\alpha_{link} \le \ln 2
\end{equation}
for Poisson-like activity patterns.
Note that these consistency conditions ensure that the algorithm works for any length $t_{tol}$.

\subsection{Construction of empirical temporal networks}
We generate an unweighted underlying topology $\mathcal{S}$ and a temporal network $\mathcal{T}$ for each empirical dataset.
We first determine the length of $\mathcal{T}$ by counting the number of timestamps in the dataset.
Then, the snapshot at time $t$ is formed by all contact events with the corresponding timestamps. 
Finally, we obtain $\mathcal{S}$ by aggregating all snapshots, that is, link $(i,j)$ exists on $\mathcal{S}$ if individuals $i$ and $j$ interact  at least once.

\newpage
\begin{center}
\Large\textbf{Supplementary Material}
\end{center}

\setcounter{equation}{0}
\setcounter{figure}{0}
\setcounter{section}{0}
\setcounter{table}{0}
\renewcommand{\thesection}{\arabic{section}}
\renewcommand{\thesubsection}{\arabic{section}.\arabic{subsection}}
\renewcommand{\theequation}{S\arabic{equation}}
\renewcommand{\thefigure}{S\arabic{figure}}
\renewcommand{\thetable}{S\arabic{table}}

\section{Renewal process} 

The activation process of a node/link can be naturally modeled as a point process.
The inter-event time (IET) of two consecutive activations is a random variable $\xi$, and we assume the IETs are independent of each other and are identically distributed with a nonnegative distribution. 
Such a point process is called a renewal process.
                                            
We set $\xi$ to be lattice with $d=1$, that is, the activation can only occur at positive integral moments, which leads to $\sum_{n=1}^{\infty} \prob(\xi=n) = 1$.
Therefore, the renewal process of a node/link is a discrete-time stochastic process $\{\Xi_n \}_{n\ge 0}$, where $\Xi_n = 1$ if the node/link is active at time $n$, otherwise $\Xi_n=0$. 

\section{Temporal network construction based on static networks}
We consider the construction of temporal networks based on a static undirected network. 
We use the probabilistic graphical model to describe the activation process of individual nodes/links in a network system at each time step and generate a stochastic process to formalize the construction process.

\subsection{Two-node systems}

\subsubsection{Theoretical analysis}
The elementary component of a network system is a two-node system, where two nodes $x$ and $y$ are connected by a link $z$ (see Fig.~1A in the main text).
The renewal processes of nodes $x$, $y$ and link $z$ are denoted as $\{X_n \}_{n\ge 0}$, $\{Y_n \}_{n\ge 0}$ and $\{Z_n \}_{n\ge 0}$.
The initial states of the nodes and the link are all active (i.e. $X_0=Y_0=Z_0=1$), and for simplicity, we omit the specification of the initial state in the rest of this Supplementary Material unless otherwise specified.  

The targeted IET distribution of node $x$ is a probability mass function $F(n)$,
\begin{equation}
	F(n) = \prob(X_0=1, X_1=0, X_2=0, ..., X_{n}=1), \quad n \in \mathbb{Z}^+,
\end{equation}
which represents the probability of the time interval between two consecutive activations. 
The probability mass function of node $y$ and link $z$ is denoted as $G(\Delta t)$ and $H(\Delta t)$, respectively. 
The trajectory of node $x$ ($y$ or $z$) until time $n$ is denoted as $\w_x^{(n)}=(w^{(0)}_{x},...,w^{(n)}_{x})^{\text{T}}$ ($\w_y^{(n)}=(w^{(0)}_{x},...,w^{(n)}_{x})^{\text{T}}$ or $\w_z^{(n)}=(w^{(0)}_{z},...,w^{(n)}_{z})^{\text{T}}$), recording all $n+1$ historical states.
We assume that link $z$ is active at time $n$ if and only if node $x$ and $y$ are all active at time $n$, which leads to 
\begin{equation*}
	\w_z^{(n)} = \w_x^{(n)}\circ \w_y^{(n)},
\end{equation*}
where the operation $\circ$ is the Hadamard product.

Given the trajectory $\w_x^{(n-1)}$, the conditional probability that node $x$ is active at time $n$ is denoted as $p_x(\w^{(n-1)}_x,1)$.
To calculate the probability $p_x(\w^{(n-1)}_x,1)$, we provide two identities by the property of renewal processes,
\begin{equation*}
\begin{aligned}
&\prob(X_n=1|X_{n-1}=0,...,X_1=0,X_0=1) = \frac{\prob(X_0=1, X_1=0, X_2=0, ..., X_{n}=1)}{\prob(X_0=1,X_1=0,...X_{n-1}=0)}= \frac{F(n)}{\sum_{i \ge n}F(i)}, \\
&\prob(X_n=a|X_m=1,X_{m-1}=w_{m-1},...,X_{1}=w_1,X_0=1)=\prob(X_{n-m}=a|X_0=1), \quad a\in \{0,1\}.	
\end{aligned}
\end{equation*}
The first equation shows the value of the conditional probability that $x$ is active at time $n$, given that $x$ is inactive at all previous moments (except the initial moment), and the second equation illustrates the state of $x$ at the current moment is determined by all historical states from the last activation moment. 

We denote the random vector $(X_0,X_1,...,X_n)^\text{T}$ by $X^{(n)}$, the conditional probability $p_x(\w^{(n-1)}_x,a)$ is given by
\begin{equation}
	p_x(\w^{(n-1)}_x,a) = \prob(X_{n}=1|X^{(n-1)}=\w^{(n-1)}_x)=
	\left \{
	\begin{aligned}
		&\frac{F(n-m)}{\sum_{i\ge n-m} F(i)},\quad &a=1, \\
		&\frac{\sum_{i\ge n-m+1}F(i)}{\sum_{i\ge n-m}F(i)}, &a=0, \\
	\end{aligned}
	\right.
\end{equation}
where $m=\max \{ k\le n: w_x^{(k)}=1 \}$ is the last activation moment of $x$. 
The same conclusions can be obtained for node $y$ and link $z$.

The construction algorithm of a two-node system is formalized by a stochastic process $\{S_n \}_{n\ge 0}$, satisfying for arbitrary sets of $t_1,...,t_k \in \mathbb{N}, k \in \Z^+$,
\begin{equation}
\begin{aligned}
	&\mu_{S_{2t_1},...,S_{2t_k}} = \mu_{X_{t_1},...,X_{t_k}}, \\
	&\mu_{S_{2t_1+1},...,S_{2t_k+1}} = \mu_{Y_{t_1},...,Y_{t_k}}, \\
	&\mu_{(S_{2t_1}\cdot S_{2t_1+1}),...,(S_{2t_k}\cdot S_{2t_k+1})} = \mu_{Z_{t_1},...,Z_{t_k}},
\end{aligned}
\end{equation} 
and $S_0=S_1=1$, where $\mu_{X_{t_1},...,X_{t_k}}$ represents the finite dimensional distribution of $\{S_n \}_{n\ge 0}$ at the time slice $(t_1,...,t_k)$.

To prove the existence of $\{S_n \}_{n\ge 0}$, we need to specify all finite dimensional distributions of $\{S_n \}_{n\ge 0}$. We assume $X_{n+1}$ and $Y^{(n)}$ be conditionally independent with respect to  $X^{(n)}$, that is, 
\begin{equation}
	\prob(X_{n+1},Y^{(n)}|X^{(n)}) = \prob(X_{n+1}|X^{(n)})\cdot \prob(Y^{(n)}|X^{(n)}),
\end{equation}
which means that the state of $x$ at the current moment is unrelated to the past of $y$.
By symmetry, we have 
\begin{equation*}
	\prob(Y_{n+1},X^{(n)}|Y^{(n)}) = \prob(Y_{n+1}|Y^{(n)})\cdot \prob(X^{(n)}|Y^{(n)}).
\end{equation*}
When $X_{n+1}=a$ and $X^{(n)}=\w^{(n)}_x$, using Eqs.~(S3) and (S4),
\begin{equation}
\begin{aligned}
 \prob (S_{2n+2}=a|S^{(2n+1)}) &= \prob(X_{n+1}=a|X^{(n)},Y^{(n)}) =\frac{\prob(X_{n+1}=a,Y^{(n)}|X^{(n)})\cdot \prob(X^{n})}{\prob(Y^{(n)}|X^{(n)})\cdot \prob(X^{n})} \\
&= \frac{\prob(X_{n+1}=a|X^{(n)})\cdot \prob(Y^{(n)}|X^{(n)})}{\prob(Y^{(n)}|X^{(n)})} =\prob(X_{n+1}=a|X^{(n)})\\
&=p_x(\w^{(n)}_x,a),
\end{aligned}
\end{equation}
and
\begin{equation}
\begin{aligned}
	\prob (S_{2n+3}=b|S^{(2n+2)}) &= \prob(Y_{n+1}=b|X^{(n+1)},Y^{(n)})= \frac{\prob(X^{(n+1)},Y^{(n+1)})}{\prob(X^{(n+1)},Y^{(n)})} \\
	&= \frac{\prob(X_{n+1}=a,Y_{n+1}=b|X^{(n)},Y^{(n)})\prob(X^{(n)},Y^{(n)})}{\prob(X_{n+1}=a|X^{(n)},Y^{(n)})\prob(X^{(n)},Y^{(n)})} \\
	&= \frac{\prob(X_{n+1}=a,Y_{n+1}=b|X^{(n)},Y^{(n)})}{p_x(\w^{(n)}_x,a)}.
\end{aligned}
\end{equation}
The numerator $\prob(X_{n+1},Y_{n+1}|X^{(n)},Y^{(n)})$ in Eq.~(S6) is related to $\w^{(n)}_x$, $\w^{(n)}_y$ and $\w^{(n)}_z$, which can be calculated as
\begin{subequations}
	\begin{align}
		\prob(X_{n+1}=1,Y_{n+1}=1|X^{(n)},Y^{(n)}) &= p_z(\w^{(n)}_z,1), \tag{S7a}\\
		\prob(X_{n+1}=1,Y_{n+1}=0|X^{(n)},Y^{(n)}) &= p_x(\w^{(n)}_x,1) - p_z(\w^{(n)}_z,1), \tag{S7b}\\
		\prob(X_{n+1}=0,Y_{n+1}=1|X^{(n)},Y^{(n)}) &= p_y(\w^{(n)}_y,1) - p_z(\w^{(n)}_z,1), \tag{S7c}\\
		\prob(X_{n+1}=0,Y_{n+1}=0|X^{(n)},Y^{(n)}) &= 1 + p_z(\w^{(n)}_z,1) - p_x(\w^{(n)}_x,1) - p_y(\w^{(n)}_y,1). \tag{S7d}
	\end{align}
\end{subequations}
When $\prob(B)=0$, we impose the conditional probability $\prob(A|B)=0$.
Using Eqs.~(S3)-(S7), we can calculate all finite dimensional distributions of $\{S_n \}_{n\ge 0}$ and verify that they satisfy the suitable consistency conditions. 
By the Kolmogorov extension theorem \cite{tao2011introduction}, a probability space $(\Omega,\mathcal{F},\prob)$ exists such that $\{S_n \}_{n\ge 0}$ is well-defined on this probability space.

When $F=G=H$, we obtain that Eqs.~(S7b) and (S7c) are equal to $0$, which indicates that the state of $x$, $y$, and $z$ is the same at any time step.
In this case, we claim that the system is synchronous.

\subsubsection{Consistency condition}
Mathematically, some distribution combinations $F,G,H$ cannot lead to a well-defined stochastic process $\{S_n \}_{n\ge 0}$.
The reason is that the values of Eqs.~(S7a)-(S7d) might be less than 0 and therefore violate the definition of probability measure.
Here, we propose a definition of distribution consistency $F, G, H$. 
\begin{definition}[Distribution consistency]
	The distributions $F(n)$, $G(n)$ and $H(n)$ are said to be consistent if the values Eqs.~(S7a)-(S7d) all belong to $[0, 1]$ for any possible $\w_x^{(n)}, \w_y^{(n)}$, $n\ge 0$.
\end{definition}
We also define the consistency of a two-node system.
\begin{definition}[Two-node system consistency]
The two-node system is said to be consistent if the targeted distribution of the two nodes and the link are consistent.
\end{definition}	
When a system is consistent, the IET distribution of each node (link) fulfills the targeted distribution.
Here, we first propose a necessary condition to verify system consistency with general distributions and analyze three classes of distributions -- power-law distributions, (discrete) exponential distributions, and Poisson distributions.
All conclusions are based on the premise that the length of temporal networks $t_{tol}$ is free.

The necessary condition is that the support set for $H$ needs to be a denumerable set, i.e., $\forall m \ge 1$, $\exists n \ge m$, $H(n)>0$.
If not, there exists a positive integer $m$, such that $H(m) > 0$ and $H(n)=0$, for all $n > m$.
When 
\begin{equation}
	\w_z^{(m-1)} = \textbf{0}^{(m-1)} = (1,\underbrace{0,...,0}_{m-1})^{\text{T}},
\end{equation}
the identity $p_z(\w_z^{(m-1)},1) = 1$ holds, meaning that link $z$ is almost surely active at time $m$.
However, it is impossible to find $F$ and $G$ such that $p_x(\w_x^{(m-1)},1) = p_y(\w_y^{(m-1)},1) = 1$ holds for all $\w_x^{(m-1)} \circ \w_y^{(m-1)} = \textbf{0}^{(m-1)}$. 
In other words, the probability $p_x(\w_x^{(m-1)},1)$ or $p_y(\w_y^{(m-1)},1)$ would be smaller than $p_z(\w_z^{(m-1)},1)$ under some trajectories, which results in Eq.~(S7b) or (S7c) being smaller than $0$.

Next, we consider that $F$, $G$ and $H$ are all power-law distributions, that is,
\begin{equation*}
	F(n) = C_1 n^{-\alpha_1},~G(n) = C_2 n^{-\alpha_2},~H(n) = C_3 n^{-\beta},
\end{equation*}
where $\alpha_1$, $\alpha_2$, $\beta$ are the exponents larger than $1$, and $C_i=(\sum_{n=1}^{\infty} n^{-\alpha_i})^{-1}$ ($i=1,2,3$) are the normalization constants.
The conditional probability of $x$ being active for the first time at time $n$ is given by
\begin{equation}
\begin{aligned}
	p_x^{(n)} = p_x(\textbf{0}^{(n-1)},1)  &= \frac{F(n)}{\sum_{i \ge n} F(i)} \approx \frac{C_1\int_{n-1/2}^{n+1/2} x^{-\alpha_1} dx}{C_1\int_{n-1/2}^{+\infty} x^{-\alpha_1} dx} \\
	&= \frac{(n-\frac{1}{2})^{-\alpha_1+1} - (n+\frac{1}{2})^{-\alpha_1+1}}{(n-\frac{1}{2})^{-\alpha_1+1}} = 1 - (\frac{n+\frac{1}{2}}{n-\frac{1}{2}})^{-\alpha_1+1}.
\end{aligned}
\end{equation} 
Equation (S9) shows that $p_x^{(n)}$ is monotonically decreasing/increasing with respect to $n$/$\alpha_1$, and when $n \to \infty$, $p_x^{(n)} \to 0$.
Since $w_x^{(i)} \ge w_z^{(i)}$ and $w_y^{(i)} \ge w_z^{(i)}$ for all $i\ge 0$, by Eqs.~(S7b) and (S7c), the exponent of the nodes need to be not less than that of the link, $\alpha_1 \ge \beta$ and $\alpha_2 \ge \beta$.
In addition, a necessary and sufficient condition for Eq.~(S7d) belonging to $[0,1]$ for all trajectories is 
\begin{equation*}
	p_x^{(1)}  + p_y^{(2)}< 1, \quad p_x^{(1)}  + p_y^{(1)} < 1 + p_z^{(1)}.
\end{equation*}
In particular, when all nodes have the same exponent $\alpha$, the first condition becomes $p_x^{(1)}  + p_x^{(2)}< 1$ and the second condition becomes $2C_1< 1 + C_3$. 

The exponential distribution we discuss is a discrete exponential distribution, and the analytical forms of $F, H, G$ are given by
\begin{equation*}
	 F(n) = C_1\int_{n-1/2}^{n+1/2} \lambda_1 e^{-\lambda_1 x} \d x,\ G(n) = C_2\int_{n-1/2}^{n+1/2} \lambda_2 e^{-\lambda_2 x} \d x,\ H(n) = C_3\int_{n-1/2}^{n+1/2} \lambda_3 e^{-\lambda_3 x} \d x,
\end{equation*}
where $\lambda_i$ ($i=1,2,3$) are the exponents and $C_i = e^{\frac{\lambda_i}{2}}$ ($i=1,2,3$) are the normalization constants. 
\begin{equation}
	p_x^{(n)} = \frac{C_1\int_{n-1/2}^{n+1/2} \lambda_1 e^{-\lambda_1 x} \d x}{C_1\int_{n-1/2}^{+\infty} \lambda_1 e^{-\lambda_1 x} \d x} = \frac{e^{-\lambda_1(n-\frac{1}{2})}-e^{-\lambda_1(n+\frac{1}{2})}}{e^{-\lambda_1(n-\frac{1}{2})}}=1-e^{-\lambda_1}. 
\end{equation}
Equation (S10) shows that $p_x^{(n)}$ is monotonically increasing with respect to $\lambda_1$ and independent to $n$. 
To satisfy the non-negativity of Eqs.~(S7a) and (S7b), the exponent of the nodes need to be not less than the link (i.e. $\lambda_1 \ge \lambda_3$, $\lambda_2 \ge \lambda_3$).
To satisfy the non-negativity of Eq.~(S7d), the necessary and sufficient condition is $e^{-\lambda_1} + e^{-\lambda_2} \ge e^{-\lambda_3}$. 
When two nodes have the same exponent, the conditions simplifies to $0\le \lambda_{node}-\lambda_{link} \le \ln 2$, which means nodes need to be more frequently active than links (the lower bound), while not too frequently (the upper bound).

Another common discrete distribution defined over $\mathbb{N}$ is the Poisson distribution. 
We prove that when $F$, $G$, and $H$ are all Poisson distributions, the system is inconsistent.
Since $n \ge 1$, we modify the definition domain of the Poisson distribution, $\prob(X=n) = e^{-\lambda} \frac{\lambda^{n-1}}{(n-1)!}$, $\forall n \in \mathbb{Z}^+$.
We have
\begin{equation}
	p_x^{(n)} = \frac{e^{-\lambda} \frac{\lambda^{n-1}}{(n-1)!}}{\sum_{i \ge n} e^{-\lambda} \frac{\lambda^{i-1}}{(i-1)!}} = \frac{1}{1 + \sum_{i\ge 1} \frac{\lambda^{i}}{(n+i)!/n!}} = \frac{1}{1 + \frac{\lambda}{n} + \frac{\lambda^2}{(n+1)n}+...}.
\end{equation}
Equation (S11) shows that $p_x^{(n)}$ is monotonically decreasing with respect to $\lambda$ and monotonically increasing with respect to $n$, when $n \to \infty$, $p_x^{(n)} \to 1$. 
Therefore, for all $\lambda_{node},\lambda_{link}$, a positive integer $K$ exists such that $p^{(K)}_H > p^{(1)}_F$, which indicates that the value of Eq.~(S7b) is less than 0.  

For a more general combination $(F, G, H)$, we need to calculate the result of Eq.~(S7) at each time step, and once one of the probabilities in Eq.~(S7) is less than 0, the construction stops, and the system is inconsistent.

\begin{algorithm}[t]
\caption{Construction on two-node systems} 
\hspace*{0.02in} {\bf Input:} 
IET distributions $F(\Delta t)$, $G(\Delta t)$, $H(\Delta t)$ and parameter $t_{tol}$\\
 \hspace*{0.02in} {\bf Output:}
trajectories $\w_x^{t_{tol}}$, $\w_y^{t_{tol}}$, $\w_z^{t_{tol}}$ or 0
\begin{algorithmic}[1]
\For{$t=1$ to $t_{tol}$}
	\State Compute $p(\w^{(t-1)}_x,1)$, $p(\w^{(t-1)}_y,1)$ and $p(\w^{(t-1)}_z,1)$
	\State Compute $p_1 = p_z(\w^{(t-1)}_z,1)$, $p_2 = p_x(\w^{(t-1)}_x,1) - p_z(\w^{(t-1)}_z,1)$, $p_3 = p_y(\w^{(t-1)}_y,1) - p_z(\w^{(t-1)}_z,1)$, and $p_4 = 1 + p_z(\w^{(t-1)}_z,1) - p_x(\w^{(t-1)}_x,1) - p_y(\w^{(t-1)}_y,1)$
	\If{$p_1 < 0$ or $p_2 < 0$ or $p_3 < 0$ or $p_4 < 0$}
		\State \Return 0
	\EndIf
	\State Choose two independent random numbers $p$, $q$ uniformly in $[0,1]$
	\If{$p < p(\w^{(t-1)}_x,1)$}
		\State $w_x^{(t)}=1$ 
		\State $r=p_1/p(\w^{(t-1)}_x,1)$
		\If{$q < r$}
			\State $w_y^{(t)}=1$ 
		\Else
			\State $w_y^{(t)}=0$
		\EndIf
	\Else
		\State $w_x^{(t)}=0$ 
		\State $r=p_3/(1-p(\w^{(t-1)}_x,1))$
		\If{$q < r$}
			\State $w_y^{(t)}=1$ 
		\Else
			\State $w_y^{(t)}=0$
		\EndIf
	\EndIf
	\State $w_z^{(t)}=w_x^{(t)}\cdot w_y^{(t)}$
\EndFor
\Return  $\w_x^{t_{tol}}$, $\w_y^{t_{tol}}$, $\w_z^{t_{tol}}$
\end{algorithmic}
\end{algorithm}

\subsubsection{Construction algorithm}
Algorithm 1 describes the procedure of two-node temporal network construction.
At each time step, we first determine the activation of $x$ which is only related to the trajectory of $x$.
Then we decide the state of $y$.
Finally, the state of $z$ is determined by the states of $x$ and $y$. 
If the return of the algorithm is $0$, then the system is inconsistent. 
Otherwise, the IET distributions of $x$, $y$ and $z$ fulfill $F$, $G$ and $H$.

\subsection{Tree systems}
A natural extension of two-node systems is tree systems (see Fig.~1B in the main text). 
Here we give a theoretical explanation and an explicit algorithm for the construction of tree structures.

\subsubsection{Theoretical analysis}
A tree system with $N$ nodes is denoted as $\G=(\V,\E)$, where $\V=\{0,...,N-1\}$ is a set of nodes and $\E=\{(i,j):i,j\in \V \}$ is a set of links.
The renewal process of node $i$ and link $(m,n)$ is denoted as $\{O_n^{(i)}\}_{n\ge 0}$ and $\{E_n^{(m,n)}\}_{n\ge 0}$, respectively.
We randomly select a node $r \in \V$ as the root of the tree and classify all other nodes according to their distance from $r$.
All nodes with $l$ steps away from $r$ compose the $l$th layer of the tree.
The number of layers is denoted as $l_{tol}$, and the number of nodes in each layer $l$ is denoted as $n_{l}$.

We define a mapping $v$ from the original set of nodes $\V$ to the new one,
\begin{equation*}
	v: j \to (l,m),
\end{equation*}
where $j \in \V $ is the original serial number, $l$ is the distance from $j$ to $\textbf{r}$ and $m$ is the number of $j$ in the layer $l$.
For example, we have $v(r)=(0,1)$.
We define a partial ordering $\preccurlyeq$ on $v(\V)$, the relation $(i_1,j_1)\preccurlyeq(i_2,j_2)$ holds when $i_1 < i_2$ or $i_1 = i_2, j_1 \le j_2$. 
Based on the relationship $\preccurlyeq$, we order the elements of $v(\V)$.
We use $\tilde{v}$ to represent this sorting,
\begin{equation*}
	\tilde{v}: v(\V) \to \V.
\end{equation*} 
The composite mapping $\hat{v}=\tilde{v} \circ v : \V \to \V$ represents the reordering of $\V$, and the result is denoted as $\widetilde{\V}$.

For the link set $\E$, we also define a reordering mapping $\hat{e}$,
\begin{equation*}
	\hat{e}: (x,y) \to  (p,q),
\end{equation*}
where $p=\hat{v}(x)$, $q=\hat{v}(y)$. 

We set $u = \hat{u}^{-1}$ and $f=\hat{e}^{-1}$. Similar to the two-node system, we construct a new stochastic process $\{T_n\}_{n\ge 0}$ that for arbitrary sets of $t_1,...,t_k \in \mathbb{N}$, $k \in \Z^+$ satisfies 
\begin{subequations}
	\begin{align}
		&\mu_{T_{Nt_1+i},...,T_{Nt_k+i}} = \mu_{O_{t_1}^{u(i)},...,O_{t_k}^{u(i)}},\quad \forall i\in \{0,...,N-1\}, \tag{S12a}\\
		&\mu_{(T_{Nt_1+i}\cdot T_{Nt_1+j}),...,(T_{Nt_k+i}\cdot T_{Nt_k+j})} = \mu_{E_{t_1}^{f(i,j)},...,E_{t_k}^{f(i,j)}},\quad \forall (i,j)\in \hat{e}(\E), \tag{S12b} \\
		& T_{k} = 1,\ \forall k \in \{0,...,N-1\}, \tag{S12c}
	\end{align}
\end{subequations}
where $\mu_{O_{t_1}^{u(i)},...,O_{t_k}^{u(i)}}$ and $\mu_{E_{t_1}^{f(i,j)},...,E_{t_k}^{f(i,j)}}$ are the finite dimensional distributions about node $u(i)\in \V$ and link $f(i,j)\in \E$ at the time slice $(t_1,...,t_k)$, respectively.

Equations (S12a)-(S12c) illustrate that the marginal distributions of $\{T_n\}_{n\ge 0}$ satisfy the corresponding distributions of single nodes or single links.
We can simply consider that $\{T_n\}_{n\ge 0}$ consists of the following sequence
\begin{equation*}
	(T_0,...,T_{N-1},...,T_{Nk},...,T_{N(k+1)-1},...) = (O_{0}^{u(0)},...,O_{0}^{u(N-1)},...,O_{k}^{u(0)},...,O_{k}^{u(N-1)},...),
\end{equation*}
and $T_{Nk+i}\cdot T_{Nk+j} = E_{k}^{f(i,j)}$.

We continue with the previous notation $T^{(Nk-1)}=(T_0,...,T_{Nk-1})$.
To prove the existence of $\{T_n\}_{n\ge 0}$, we specify all finite dimensional distributions of $\{T_n\}_{n\ge 0}$.
We assume the following conditional independence,
for all $k \in \mathbb{N}$,
\begin{equation}
	\prob(T_{Nk + i}|T^{(Nk-1)})=\prob(O_k^{u(i)}|O_{k-1}^{u(i)},...,O_{0}^{u(i)}),
\end{equation}
and for arbitrary mutually unequal $i_1,...,i_n,i,m \in \mathbb{N}$ that satisfy: (1) $(i_c,i) \notin \hat{e}(\E)$, $ \forall c\in \{1,...,k\}$, (2) $(i, m) \in \hat{e}(\E)$, we have 
\begin{equation}
	\prob(T_{Nk+i}|T_{Nk+m},T_{Nk+i_1},...,T_{Nk+i_n},T^{(Nk-1)}) = \prob(O_{k}^{u(i)}|O_{k-1}^{u(i)},...,O_{0}^{u(i)},O_{k}^{u(m)},...,O_{0}^{u(m)}).
\end{equation}
Equation (S13) illustrates that, given the historical trajectory of the other nodes, the activation of node $i$ at the current state is only related to its own trajectory.
Equation (S14) indicates that the activation of node $i$ only depends on the two-node system in which it is located, given the trajectory of another node in the same two-node system. 

A corollary of Eq.~(S14) is that  
for arbitrary mutually unequal $i_1,...,i_n,i,j,m \in \mathbb{N}$ that satisfy: (1) $(i_c,i),(i_c,j) \notin \hat{e}(\E)$, $\forall c\in \{1,...,k\}$, (2) $(i,m),(j,m) \in \hat{e}(\E)$, 
\begin{equation}
	\begin{aligned}
		\prob(T_{Nk+i},T_{Nk+j}|T_{Nk+m},T_{Nk+i_1},...,T_{Nk+i_n},T^{(Nk-1)}) &= \prob(O_{k}^{u(i)}|O_{k-1}^{u(i)},...,O_{0}^{u(i)},O_{k}^{u(m)},...,O_{0}^{u(m)}) \\
		&\times \prob(O_{k}^{u(j)}|O_{k-1}^{u(j)},...,O_{0}^{u(j)},O_{k}^{u(m)},...,O_{0}^{u(m)}).
	\end{aligned}
\end{equation}
Equation (S15) demonstrates the conditional independence for the triplets (r,a1,b1) and (r,a1,a2) shown in Fig.~1B in the main text. 

Using Eqs.~(S12)-(S14), we can calculate all finite dimensional distributions of $\{T_n\}_{n\ge 0}$, then the existence of $\{T_n\}_{n\ge 0}$ can be proved with the Kolmogorov extension theorem \cite{tao2011introduction}.

\subsubsection{Consistency condition}
We propose the definition of the consistency of a tree system.
\begin{definition}[Tree system consistency]
	A tree system is said to be consistent if all two-node systems in the tree system are consistent.
\end{definition}
This definition shows that the consistency of a tree system is determined by the consistency of each two-node system.

Similar to two-node systems, we also have equivalents for the system consistency under different activity patterns.
When the targeted distributions are all power-law distributions, the necessary and sufficient condition is that each two-node system satisfies $p_x^{(1)}  + p_y^{(2)}< 1$, $p_x^{(1)}  + p_y^{(1)} < 1 + p_z^{(1)}$, and the exponents of nodes are larger than the link.
When the targeted distributions are discrete exponential distributions, the necessary and sufficient condition is that the difference between the exponents of each node and each link  is between 0 and $\ln 2$ in each two-node system.
When the targeted distributions are Poisson distributions, the tree system is inconsistent.

\begin{algorithm}[t]
\caption{Construction on tree systems} 
\hspace*{0.02in} {\bf Input:} 
tree structure $\mathcal{T}$, IET distributions for all nodes and links, and parameter $t_{tol}$\\
 \hspace*{0.02in} {\bf Output:}
trajectories of all nodes and links
\begin{algorithmic}[1]
\For{$t=1$ to $t_{tol}$}
	\State Select the root $r$ and label all other nodes as $(l,m)$ according to the distance from $r$
	\State $p_{\textbf{{r}}}=p(\w^{(t-1)}_{(0,1)},1)$ ($l=0$)
	\State Select a random number $q$ uniformly in $[0,1]$
	\If{$q < p_{\textbf{{r}}}$}
		\State $w_{\textbf{r}}^{(t)}=1$ 
	\Else
		\State $w_{\textbf{r}}^{(t)}=0$
	\EndIf
	\For{$l=1$ to $l_{tol}$}
		\For{$m=1$ to $n_l$}
			\State Search $k$, such that $x=v^{-1}(l-1,k)$, $y=v^{-1}(l,m)$, $z=(x,y)\in \E$
			\State Execute a single loop of Algorithm 1 on the two-node system consisting of
			\State the nodes $x,y$ and the link $z$
		\EndFor		
	\EndFor	
\EndFor
\Return  trajectories of all nodes and links
\end{algorithmic}
\end{algorithm}

\subsubsection{Construction algorithm}
Algorithm 2 shows the construction of temporal networks on tree systems. 
At each time step, the state of the root $r$ is updated according to its own trajectory.
Then the state of leaves is updated sequentially by executing Algorithm 1 on every two-node system. 
For each loop in lines 10-13, the state of all nodes on layer $l$ is updated. 
Specifically, to determine the state of $y$ on layer $l$, we first identify the node $x$ on layer $l-1$, which is directly connected to $y$.
At this point, the state of $x$ has already been updated.
Then, the state of $y$ is updated within the two-node system formed by $x$ and $y$.

\subsection{Arbitrary structured systems}
We can always find a spanning tree for any network structure.
We call a link a \emph{trunk} if it is in the spanning tree, a \emph{branch} otherwise.
The activity of a network system is determined by its spanning tree.
Algorithm 3 shows the construction of temporal networks on arbitrary underlying topologies $\G$.
In the algorithm, we first select a spanning tree of $\G$.
At each time step, we execute Algorithm 2 on the spanning tree, so that the states of all nodes and trunks are updated.
Then, the state of each branch is determined according to the state of the nodes on both ends.

We claim that the system is consistent if and only if its spanning tree system is consistent.
In particular, when all nodes and trunks have the same targeted IET distribution, the whole system is synchronous, and therefore all nodes and links (including branches) always become active or inactive at the same time.
This property ensures that, for any targeted distributions of nodes and links (fulfilling the consistency condition), we can always find inputs for nodes and trunks, such that the algorithmic IET distributions match with the target ones. 

\begin{algorithm}[t]
\caption{Construction based on arbitrary network systems} 
\hspace*{0.02in} {\bf Input:} 
network $\G$, IET distributions of all nodes and trunks, and parameter $t_{tol}$ \\
 \hspace*{0.02in} {\bf Output:} 
 trajectories of all nodes and links
\begin{algorithmic}[1]
\State Select a spanning tree $\mathcal{T}$ of $\G$
\State Assign a probability mass function to each node and each link in $\mathcal{T}$
\For{$t=1$ to $t_{tol}$}  
	\State Execute a single loop of Algorithm 2 for $\mathcal{T}$
	\State Update the state of each branch according to the state of the nodes on both sides
\EndFor
\Return  trajectories of all nodes and links
\end{algorithmic}
\end{algorithm}

\subsection{Systems with ring structure}
A significant property of a tree system $\mathcal{T}=(\V, \E)$ is that there is no ring structure, that is, for all $i_1,...,i_m \in \V$, $a_{i_1 i_2}\times ... \times a_{i_{m-1} i_{m}} \times a_{i_{m}i_1} = 0$, where $a_{ij}=1$ if $(i,j)\in \E$, otherwise $a_{ij}=0$.
For a general network structure, the ring structure exists and one of the links in the ring would be a branch.

Below, we analyze the activity of branches.  
We consider the simplest system with a ring, which consists of three nodes, $\{x,y,z\}$, and three links, $\{(x,y),(y,z),(z,x)\}$.
The renewal process of nodes $x$, $y$ and $z$ is denoted as $\{X_n \}_{n\ge 0}$, $\{Y_n \}_{n\ge 0}$ and $\{Z_n \}_{n\ge 0}$.
We set $x$ as the root $r$.
At time $n+1$, the probability of $x$ being active is $p (\w^{(n)}_x,1)$ and the result is denoted as $a$.
To make the IET distribution of node $y$ and link $(x,y)$ fulfill targeted distributions, the probability of $y$ being active is 
\begin{equation*}
	\frac{\prob(X_{n+1}=a,Y_{n+1}=1|X^{(n)}=\w^{(n)}_x,Y^{(n)}=\w^{(n)}_y)}{p (\w^{(n)}_x,a)}.
\end{equation*}
This result is denoted as $b$.
Then, we can calculate the probability of $z$ being active through the structure $\{y,(y,z), z\}$ or $\{x,(x,z), z\}$.
The former leads to the probability
\begin{equation}
	\frac{\prob(Y_{n+1}=b,Z_{n+1}=1|Y^{(n)}=\w^{(n)}_y,Z^{(n)}=\w^{(n)}_z)}{p (\w^{(n)}_y,b)},
\end{equation}
and the latter leads to the probability
\begin{equation}
	\frac{\prob(X_{n+1}=a,Z_{n+1}=1|X^{(n)}=\w^{(n)}_x,Z^{(n)}=\w^{(n)}_z)}{p (\w^{(n)}_x,a)}.
\end{equation}
These two probabilities are not always the same.
Using Eq.~(S16) (Eq.~(S17)) to obtain the state of $z$ is equivalent to the select the spanning tree with links $(x,y), (y,z)$ ($(x,y), (x,z)$), and the state of the remaining link (i.e. the branch) is then established.

Furthermore, in general, the activity of branches is not a (strict) renewal process.
Without loss of generality, we assume that links $(x,y), (x,z)$ are in the spanning tree.
We denote the stochastic process of link $(y,z)$ as $\{S_n \}_{n\ge 0}$.
To prove the above assertion, we verify a sufficient condition
\begin{equation}
	\prob(S_2=1|S_1=1) \neq \prob(S_1=1).
\end{equation}
The left-hand side of Eq.~(S18) equals 
\begin{equation}
	\frac{1}{\prob(S_1=1)}\sum_{x_1,x_2\in \{0, 1\}} \prob(X_1=x_1,Y_1=1,Z_1=1,X_2=x_2,Y_2=1,Z_2=1).
\end{equation}
The numerator of Eq.~(S19) equals 
\begin{equation}
\begin{aligned}
	\sum_{x_1\in \{0,1\}} \prob(X_1=x_1,Y_1=1,Z_1=1) \sum_{x_2\in \{0,1\}} \left(\substack{\prob(X_2=x_2,Y_2=1|X_1=x_1,Y_1=1)  \\
	\times \prob(X_2=x_2,Z_2=1|X_1=x_1,Z_1=1)} \right) / \prob(X_2=x_2|X_1=x_1)
\end{aligned}
	.
\end{equation}
In general, Eq.~(S20) is not equal to 
\begin{equation*}
	\sum_{x_1\in \{0,1\}} \prob(X_1=x_1,Y_1=1,Z_1=1) \sum_{x_2\in \{0,1\}} \prob(X_1=x_2,Y_1=1,Z_1=1) = \prob(S_1=1)^2,
\end{equation*}
which means that Eq.~(S18) holds.
However, we can still use renewal processes to approximate the activity of branches (see fig.~S6 and section 4 for reasons)

\section{Robustness analysis}
In the main text, we find that the aggregated IET distributions of nodes and links are robust to underlying topologies and the selection of spanning trees.
The main reason is that the IET distribution of branches is only sensitive to targeted distributions.
Specifically, we prove that the algorithmic distribution of every single branch is an exponential distribution (a heavy-tailed distribution) in Poisson-like activity patterns (bursty activity patterns) with upper and lower bounds determined by algorithm inputs.
Here, we provide theoretical explanations.
\subsection{Activity of branches in Poisson-like activity patterns}
We begin our analysis with a three-node ring consisting of three nodes $x,y,z$ and three links $(x,y)$,$(y,z)$,$(z,x)$.
We select a spanning tree with links $(x,y)$ and $(y,z)$, and the root is node $x$.
The renewal process of nodes $x,y,z$ is denoted as $\{X_n \}_{n\ge 0}$, $\{Y_n \}_{n\ge 0}, \{Z_n \}_{n\ge 0}$, and the exponent of nodes $x,y,z$ (links $(x,y), (x,z)$) is $\lambda_x, \lambda_y, \lambda_z$ ($\lambda_{xy}, \lambda_{xz}$).
For simplicity, we make the following assumption:

\emph{Assumption 1:} $\lambda_x=\lambda_y=\lambda_z=\lambda_1 > 0$ and $\lambda_{xy}=\lambda_{xz}=\lambda_2 > 0$.

Considering the consistency condition for Poisson-like activity patterns, the relation $0 \le \lambda_1-\lambda_2 \le \ln 2$ holds.
The stochastic process of link $(x,z)$ is denoted as $\{S_n \}_{n\ge 0}$.

We define a random variable
\begin{equation*}
	\tau_{x} = \inf\{n\ge 1:X_n=1\},
\end{equation*}
which is a stopping time for node $x$, indicating the first activation time of $x$.
Similarly, the stopping times for nodes $y,z$ and link $(x,z)$ are denoted as $\tau_y$, $\tau_z$ and $\tau_s$, respectively.
We approximate the activity of $(x,z)$ by a renewal process with the distribution of $\tau_s$ as its IET distribution.

We first prove the following theorem,
\begin{theorem}
	For a three-node ring satisfying Assumption 1, 
	\begin{equation*}
		\prob(\tau_s > n) = e^{-\mu n},
	\end{equation*}
	where $\mu > 0$.
\end{theorem}
If Theorem 1 holds, the algorithmic IET distribution of the branch is also exponential.
\begin{proof}
	By the definition of $\tau_s$, we have
	\begin{equation*}
	\begin{aligned}
		\prob(\tau_s > n) &= \sum_{\a^{(n)} \circ \c^{(n)} = \textbf{0}^{(n)}} \prob(X^{(n)} =\a^{(n)}, Z^{(n)}=\c^{(n)}) \\
		&= \sum_{\a^{(n)} \circ \c^{(n)} = \textbf{0}^{(n)}} \sum_{\b^{(n)} \in \{0, 1\}^{n+1}} \prob(X^{(n)}=\a^{(n)}, Y^{(n)}=\b^{(n)}, Z^{(n)}=\c^{(n)}),
	\end{aligned}
	\end{equation*} 
	where $\a^{(n)}=(a_0,...,a_n)^{\text{T}}$, $\b^{(n)}=(b_0,...,b_n)^{\text{T}}$, and $\c^{(n)}=(c_0,...,c_n)^{\text{T}}$  are the trajectories of $x,y,z$ with length $n+1$, the definition of $\textbf{0}^{(n)}$ is the same as that in  Eq.~(S8), and the operation $\circ$ is the Hadamard product.
	
	For each fixed $\a^{(n)},\b^{(n)},\c^{(n)}$, using Eqs.~(S13) and (S14), we have 
\begin{equation}
\begin{aligned}
		\prob(X^{(n)}=\a^{(n)}, Y^{(n)}=\b^{(n)}, Z^{(n)}&=\c^{(n)})= \prob(X_n=a_n|X^{(n-1)}=\a^{(n-1)})\\
		&\times \prob(Y_n=b_n|X^{(n)}=\a^{(n)}, Y^{(n-1)}=\b^{(n-1)})\\
		&\times \prob(Z_n=c_n|Y^{(n)}=\b^{(n)}, Z^{(n-1)}=\c^{(n-1)})\\
		&\times \prob(X^{(n-1)}=\a^{(n-1)}, Y^{(n-1)}=\b^{(n-1)}, Z^{(n-1)}=\c^{(n-1)}).
\end{aligned}
\end{equation}
This gives
\begin{equation}
	\begin{aligned}
			\prob(\tau_s > n) &= \sum_{\a^{(n-1)} \circ \c^{(n-1)} = \textbf{0}^{(n-1)}} \sum_{\b^{(n-1)} \in \{0, 1\}^{n}}  \prob(X^{(n-1)}=\a^{(n-1)}, Y^{(n-1)}=\b^{(n-1)}, Z^{(n-1)}=\c^{(n-1)}) \\
			&\times \sum_{a_n \times c_n = 0} \sum_{b_n \in \{0, 1\}} \left( \prob(X_n=a_n|X^{(n-1)}=\a^{(n-1)}) \right.\\
			&\times \prob(Y_n=b_n|X^{(n)}=\a^{(n)}, Y^{(n-1)}=\b^{(n-1)}) \\
			&\left. \times \prob(Z_n=c_n|Y^{(n)}=\b^{(n)}, Z^{(n-1)}=\c^{(n-1)}) \right).
	\end{aligned}
\end{equation}
Using Eqs.~(S5)-(S7) and (S10), we have
\begin{equation}
	\begin{aligned}
		\sum_{a_n \times c_n = 0} \sum_{b_n \in \{0, 1\}} &\prob(X_n=a_n|X^{(n-1)}=\a^{(n-1)}) \\
		&\times \prob(Y_n=b_n|X^{(n)}=\a^{(n)}, Y^{(n-1)}=\b^{(n-1)}) \\
		 &\times \prob(Z_n=c_n|X^{(n)}=\c^{(n)}, Z^{(n-1)}=\b^{(n-1)}) \\
		 &= e^{-\lambda_1} + \frac{(1-e^{-\lambda_2})(e^{-\lambda_2} - e^{-\lambda_1})}{1-e^{-\lambda_1}} + \frac{(e^{-\lambda_2} - e^{-\lambda_1})(2e^{-\lambda_1} - e^{-\lambda_2})}{e^{-\lambda_1}} \\
		 &:= C(\lambda_1, \lambda_2). 
	\end{aligned}
\end{equation} 
This leads to
\begin{equation*}
	\prob(\tau_s > n) = C \prob(\tau_s > n-1).
\end{equation*}

Since $\prob(\tau_s > 0)$ = 1, we have
\begin{equation*}
	\prob(\tau_s > n) = C^n.
\end{equation*}

In this simplest ring, we can also analyze the decay of $\prob(\tau_s > n)$ in three cases.

Case 1: When $\lambda_1 = \lambda_2$, we have $C = e^{-\lambda_1}$, which indicates that 
\begin{equation*}
	\prob(\tau_s > n) = e^{-\lambda_1 n} = \prob(X^{(n)}=\textbf{0}^{(n)}, Y^{(n)}=\textbf{0}^{(n)}, Z^{(n)}=\textbf{0}^{(n)}).
\end{equation*} 
This identity can also be obtained by the synchronization of the system because all elements have the same targeted distribution.

Case 2: When $\lambda_1 - \lambda_2=\ln 2$, we have
\begin{equation*}
	C=\frac{e^{-\lambda_2}}{2} + \frac{(1-e^{-\lambda_2})e^{-\lambda_2}}{2-e^{-\lambda_2}} < 1,
\end{equation*}
meaning that the decay of $\prob(\tau_s > n)$ is exponential.

Let $C = e^{-\mu}$, 
we compare the magnitude of $\mu(\lambda_2)$ and $\lambda_2$. 
Since
\begin{equation*}
	\frac{\d (\mu(\lambda_2)-\lambda_2)}{\d \lambda_2} = -\frac{2 e^{\lambda_2}}{8e^{2\lambda_2} - 10 e^{\lambda_2} + 3} < 0.
\end{equation*}
We get 
\begin{equation*}
	\mu(\lambda_2)-\lambda_2 > 0, \quad \mu(\lambda_2) \to \lambda_2 ~\text{as}~ \lambda_2 \to \infty.
\end{equation*}
This indicates that $\prob(\tau_s > n)$ is also an exponential distribution, and the corresponding exponent is larger than that of single trunks, $\lambda_2$.

Case 3: When $0 < \lambda_1 - \lambda_2 < \ln 2$, let $x = e^{-\lambda_1}, y=e^{-\lambda_2}$, then
\begin{equation*}
	\partial_y C(x,y) = \frac{2(x^2-2x+y)}{x(x-1)}.
\end{equation*}
This gives
\begin{equation*}
	\max_{y} C(x,y) = (2-x)x < 1.
\end{equation*}

Similar to Case 2, let $C = e^{-\mu}$, 
\begin{equation}
	\frac{\d (\mu(\lambda_2)-\lambda_2)}{\d \lambda_2} = \frac{e^{3\lambda_1} + e^{2\lambda_2} - 2e^{\lambda_1 + 2\lambda_2}}{e^{3\lambda_1} - e^{2\lambda_2}+2e^{\lambda_1+\lambda_2} - 4e^{2\lambda_1 + \lambda_2}+2e^{\lambda_1+2\lambda_2}}.
\end{equation}
For each fixed $\lambda_1$, the numerator and denominator of Eq.~(S24) are quadratic functions with respect to $e^{\lambda_2}$.
Considering $e^{\lambda_1}/2<e^{\lambda_2}<e^{\lambda_1}$, when $\lambda_1 > \ln(3/2)$, 
\begin{equation*}
	e^{3\lambda_1} - e^{2\lambda_2}+2e^{\lambda_1+\lambda_2} - 4e^{2\lambda_1 + \lambda_2}+2e^{\lambda_1+2\lambda_2} < 0.
\end{equation*}
Then 
\begin{equation*}
	\max_{\lambda_2} (\mu(\lambda_2)-\lambda_2) < \max\{\mu(\lambda_1)-\lambda_1, \mu(\lambda_1-\ln 2)-\lambda_1 + \ln 2\} = 0.
\end{equation*}
This indicates that the exponent of $\prob(\tau_s > n)$ is lower than $\lambda_2$.

\end{proof}

For a general ring $\mathcal{R}=(\V, \E)$, $\V= \{1,...,m\}$ and $\E=\{(i,j): i,j\in \V, |i-j|=1\} \cup \{(1, m) \}$.
The renewal process of individual nodes in $\V$ are denoted as $\{O_n^{(1)} \}_{n\ge 0}$,$...$, $\{O_n^{(m)}\}_{n\ge 0}$.
We add all links into the spanning tree except link $(1,m)$, and we denote the stochastic process of  $(1,m)$ by $\{S_n \}_{n\ge 0}$. 
The first activation time of branch $(1,m)$ is also denoted as $\tau_s$.
We can prove the following theorem 
\begin{theorem}
	For a general ring with $m$ nodes, if every single node (trunk) has the same exponent, then
	\begin{equation*}
		\prob(\tau_s > n) = e^{-\mu n},
	\end{equation*}
	where $\mu > 0$.
\end{theorem}
\begin{proof}
	We denote the exponent of every single node (trunk) as $\lambda_1$ ($\lambda_2$). 
Let $O^{(j,n)}$ and $\i^{(k,n)}$ represent $(O_0^{(j)},...,O_n^{(j)})$ and $(i_0^{(k)},...,i_n^{(k)})$, respectively.
We have
\begin{equation*}
\begin{aligned}
	\prob(\tau_s > n) &= \sum_{\i^{(1,n)} \circ \i^{(m,n)} = \textbf{0}^{(n)}} \prob(O^{(1,n)} =\i^{(1,n)}, O^{(m,n)}=\i^{(m,n)}) \\
	&= \sum_{\i^{(1,n)} \circ \i^{(m,n)} = \textbf{0}^{(n)}} \sum_{ \substack{\i^{(2,n)},...,\i^{(m-1,n)} \\ \in \{0,1\}^{n+1}}} \prob(O^{(1,n)} =\i^{(1,n)}, ..., O^{(m,n)}=\i^{(m,n)}).
\end{aligned}
\end{equation*}

For each fixed $\i^{(1,n)},...,\i^{(m,n)}$, similar to Eq.~(S21), we have 
\begin{equation}
	\begin{aligned}
		\prob(O^{(1,n)} =\i^{(1,n)}, ..., O^{(m,n)}=\i^{(m,n)}) &= \prob(O^{(1)}_n =i_{n}^{(1)}|O^{(1,n-1)} = \i^{(1,n-1)}) \\
		&\times \prob(O_n^{(2)}=i_n^{(2)}|O^{(1,n)}=\i^{(1,n)},O^{(2,n-1)}=\i^{(2,n-1)}) \\
		&\times ... \\
		&\times \prob(O_n^{(m)}=i_n^{(m)}|O^{(m-1,n)}=\i^{(m-1,n)},O^{(m,n-1)}=\i^{(m,n-1)}) \\
		&\times \prob(O^{(1,n-1)} =\i^{(1,n-1)}, ..., O^{(m,n-1)}=\i^{(m,n-1)}),
	\end{aligned}
\end{equation}
meaning that node $1$ is the root of the spanning tree.
This gives
\begin{equation}
	\begin{aligned}
			\prob(\tau_s > n) &= \sum_{\substack{\i^{(1,n-1)} \\ \circ  \i^{(m,n-1)} = \textbf{0}^{(n-1)}}} \sum_{\substack{\i^{(2,n-1)},...,\i^{(m-1,n-1)} \\ \in \{0,1\}^{n}}} \prob(O^{(1,n-1)} =\i^{(1,n-1)}, ..., O^{(m,n-1)}=\i^{(m,n-1)}) \\
			&\times \sum_{i_n^{(1)} \times i_n^{(m)} = 0} \sum_{i_n^{(2)},..., i_n^{(m-1)} \in \{0, 1\}} \prob(O^{(1)}_n =i_{n}^{(1)}|O^{(1,n-1)} = \i^{(1,n-1)}) \\
			&\times \prob(O_n^{(2)}=i_n^{(2)}|O^{(1,n)}=\i^{(1,n)},O^{(2,n-1)}=\i^{(2,n-1)}) \\
			&\times ... \\
			&\times \prob(O_n^{(m)}=i_n^{(m)}|O^{(m-1,n)}=\i^{(m-1,n)},O^{(m,n-1)}=\i^{(m,n-1)}). \\
	\end{aligned}
\end{equation}
Using Eqs.~(S5)-(S7) and (S10), we have
\begin{equation}
	\begin{aligned}
		\sum_{i_n^{(1)} \times i_n^{(m)} = 0} \sum_{i_n^{(2)},..., i_n^{(m-1)} \in \{0, 1\}} &\prob(O^{(1)}_n =i_{n}^{(1)}|O^{(1,n-1)} = \i^{(1,n-1)}) \\
			&\times \prob(O_n^{(2)}=i_n^{(2)}|O^{(1,n)}=\i^{(1,n)},O^{(2,n-1)}=\i^{(2,n-1)}) \\
			&\times ... \\
			&\times \prob(O_n^{(m)}=i_n^{(m)}|O^{(m-1,n)}=\i^{(m-1,n)},O^{(m,n-1)}=\i^{(m,n-1)}) \\
			&= 2e^{-\lambda_1} - \sum_{i^{(2)},..., i^{(m-1)} \in \{0, 1\}} \prob(O_1^{(1)}=0,O_1^{(2)}=i^{(2)},...,O_1^{(m-1)}=i^{(m-1)}, O_1^{(m)}=0) \\
			&:= C.
	\end{aligned}
\end{equation}
This leads to
\begin{equation*}
	\prob(\tau_s > n) = C^n.
\end{equation*}
Since $\prob(\tau_s > n) < \prob(\tau_s > n - 1) < 1$ for all $n \ge 2$, we have $C < 1$, showing that $\prob(\tau_s > n)$ is  an exponential distribution. 
\end{proof}

Furthermore,
the definition of $\tau_s$ implies that 
\begin{equation*}
	\{\tau_1 > n \} \subset \{\tau_s > n\},\quad \forall n \in \mathbb{N},
\end{equation*}
where $\tau_1$ is the first activation time of node $1$.
Since 
\begin{equation*}
	\prob (\tau_1 > n) = e^{-\lambda_1 n}.
\end{equation*} 
This gives 
\begin{equation*}
	\prob(\tau_s > n) \ge e^{-\lambda_1 n},
\end{equation*}
meaning that the distribution, $\prob(\tau_s > n)$, for any single branch has a uniform lower bounded.
And we notice that the constant $C$ in Eq.~(S27) is monotonically increasing with respect to the size of the ring, $m$, indicating that the exponent for a branch located at a large ring is smaller than that of a branch located at a small ring.
Therefore, we can find an exponential distribution that is a uniform upper bound for every single branch, and the exponent of the upper bound is related to the targeted exponent of nodes and trunks.

\subsection{Activity of branches in bursty activity patterns}
In this section, we prove that the distribution of the first activation time for every single branch has uniform upper and lower bounds, which are all heavy-tailed.

Similar to Poisson-like activity patterns, for a ring with $m$ nodes, we also select all links except link $(1,m)$ as trunks. 
The exponent of every single node (trunk) is the same, denoted as $\alpha_1$ ($\alpha_2$).

First, we offer the mathematical definition of the heavy-tailed distribution \cite{rolski2009stochastic}. 
\begin{definition}[Heavy-tailed distribution]
	The distribution of a random variable $X$ is said to have a heavy tail if 
	\begin{equation*}
		\lim_{x\to \infty} e^{\lambda x} \prob(X > x) = \infty, \quad \forall \lambda > 0.
	\end{equation*}
\end{definition}

By Definition 4, we propose the definition of a sequence to be heavy-tailed.
\begin{definition}[Heavy-tailed sequence]
	A non-negative sequence $\{a_n\}_{n=1}^{\infty}$ is said to be heavy-tailed if
	\begin{equation*}
		\lim_{n \to \infty} e^{\lambda n}a_n = \infty,  \quad \forall \lambda > 0.
	\end{equation*}
\end{definition}

Under the above definitions, we have the following lemmas.
\begin{lemma}
	For a discrete non-negative random variable $X$, let $a_n = \prob(X > n)$, if
	\begin{equation}
		\lim_{n\to \infty} \frac{a_n}{a_{n-1}} = 1,
	\end{equation}
	then $X$ is heavy-tailed.
\end{lemma}

\begin{proof}
	By Eq.~(S28), we have
	\begin{equation*}
		\lim_{n \to \infty} \frac{e^{\lambda n}a_n}{e^{\lambda (n-1)}a_{n-1}} = e^{\lambda} > 1, \quad \forall \lambda > 0.
	\end{equation*}
	Therefore,
	\begin{equation*}
		\lim_{n \to \infty} e^{\lambda n}a_n = \infty,
	\end{equation*}
	meaning that the distribution of $X$ is heavy-tailed.
\end{proof}
One straightforward implication is that the power-law distribution is heavy-tailed.

Another lemma shows that a sequence is heavy-tailed when it is lower-bounded by another heavy-tailed sequence,
\begin{lemma}
	For two non-negative sequences $\{a_n\}_{n=1}^{\infty}$ and $\{b_n\}_{n=1}^{\infty}$ that satisfy $a_0=b_0=1$ and $a_n\ge b_n$ for all $n\ge 1$, if 
\begin{equation*}
	\lim_{n\to \infty} \frac{b_n}{b_{n-1}} = 1,
\end{equation*}
then the sequence $\{a_n\}_{n=1}^{\infty}$ is heavy-tailed.
\end{lemma}

\begin{proof}
	To see this, 
we assume that there exists $\lambda_0 > 0$ such that $\lim_{n \to \infty} e^{\lambda_0 n} a_n \neq \infty$.
Then $\liminf_{n \to \infty} e^{\lambda_0 n} a_n = c < \infty$, which shows there exists a subsequence $\{a_{n_k}\}_{k=1}^{\infty}$ such that the decay of $a_{n_k}$ is exponential, and 
\begin{equation*}
	\begin{aligned}
		&\ln a_{n_k} := \sum_{i=1}^{n_k} x_i,\quad x_{n_k} = \ln \frac{a_{n_k}}{a_{{n_{k-1}}}} \to -\lambda_0 \quad \text{as}~n_{k} \to \infty, \\
		&\ln b_{n_k} := \sum_{i=1}^{n_k} y_i,\quad y_{n_k} = \ln \frac{b_{n_k}}{b_{{n_{k-1}}}} \to 0 \quad \text{as}~n_{k} \to \infty.
	\end{aligned}
\end{equation*}
Thus there exists $n_j$ such that for all $n_k \ge n_j$, $a_{n_k} < b_{n_k}$, in contradiction with the assumption that $a_n \ge b_n$.
\end{proof}

Since
\begin{equation*}
	\prob(\tau_s > n) \ge \prob(\tau_1 > n) \approx C n^{-\alpha_1 + 1},
\end{equation*} 
with the help of lemmas above, we obtain that $\prob(\tau_s > n)$ is a heavy-tailed distribution and is lower bounded by a power-law distribution with an exponent $\alpha_1$.

We next turn to derive the uniform upper bounds for branches.
\begin{theorem}
	For a general ring with $m$ nodes, if the exponent of every single node (trunk) is the same, then 
	\begin{equation*}
		A_n \le \prob(\tau_s > n) \le B_n
	\end{equation*}
	holds for any $m$, where the sequences $\{A_n\}_{n=1}^{\infty}$ and $\{B_n\}_{n=1}^{\infty}$ are both heavy-tailed.
\end{theorem}

\begin{proof}
	Similar to Eq.~(S26), for each fixed $\i^{(1,n)},...,\i^{(m,n)}$, we have 
	\begin{equation*}
		\begin{aligned}
			\prob(\tau_s > n) &= \sum_{\substack{\i^{(1,n-1)} \\ \circ  \i^{(m,n-1)} = \textbf{0}^{(n-1)}}} \sum_{\substack{\i^{(2,n-1)},...,\i^{(m-1,n-1)} \\ \in \{0,1\}^{n}}} \prob(O^{(1,n-1)} =\i^{(1,n-1)}, ..., O^{(m,n-1)}=\i^{(m,n-1)}) \\
			&\left[1 - p_1(\i^{(1,n-1)},1) + \sum_{i^{(2)},..., i^{(m-1)} \in \{0, 1\}} \prob \left(\substack{O_n^{(1)}=1,O_n^{(2)}=i^{(2)}, \\
			...,\\
			O_n^{(m-1)}=i^{(m-1)}, O_n^{(m)}=0| \\
			O^{(1,n-1)} =\i^{(1,n-1)}, ..., O^{(m,n-1)}=\i^{(m,n-1)}} \right) \right]\\
			&= \sum_{\substack{\i^{(1,n-1)} \\ \circ  \i^{(m,n-1)} = \textbf{0}^{(n-1)}}} \sum_{\substack{\i^{(2,n-1)},...,\i^{(m-1,n-1)} \\ \in \{0,1\}^{n}}} \prob(O^{(1,n-1)} =\i^{(1,n-1)}, ..., O^{(m,n-1)}=\i^{(m,n-1)}) \\
			&\left[1 - p_1(\i^{(1,n-1)},1) + 1 - p_m(\i^{(m,n-1)},1) - \sum_{i^{(2)},..., i^{(m-1)} \in \{0, 1\}} \prob \left(\substack{O_n^{(1)}=0,O_n^{(2)}=i^{(2)}, \\
			...,\\
			O_n^{(m-1)}=i^{(m-1)}, O_n^{(m)}=0| \\
			O^{(1,n-1)} =\i^{(1,n-1)}, ..., O^{(m,n-1)}=\i^{(m,n-1)}} \right) \right]\\ 
		\end{aligned}
	\end{equation*}
Let 
\begin{equation*}
	C_n^{(m)} = \sum_{\substack{\i^{(1,n-1)} \\ \circ  \i^{(m,n-1)} = \textbf{0}^{(n-1)}}} \sum_{\substack{\i^{(2,n-1)},...,\i^{(m-1,n-1)} \\ \in \{0,1\}^{n}}} \prob(O^{(1,n-1)} =\i^{(1,n-1)}, ..., O^{(m,n-1)}=\i^{(m,n-1)}) (1 - p_1(\i^{(1,n-1)},1)),
\end{equation*}
Since the statuses of nodes $1$ and $m$ are symmetric.
 We have
\begin{equation*}
	C_n^{(m)} \le \prob(\tau_s > n) \le 2C_n^{(m)},
\end{equation*}
Since $\prob(\tau_s > n)$ is heavy-tailed, using Lemma 2, we have $2C_n^{(m)}$ is heavy-tailed when $n$ is sufficiently large.
Then $\prob(\tau_s > n)$ is upper and lower bounded by heavy-tailed distributions.
Let $A_n = \min_{m}C_{n}^{(m)}$ and $B_n = \max_{m}2C_{n}^{(m)}$, $A_n$ and $B_n$ are also heavy-tailed and are only related to the targeted exponent of nodes and trunks.
\end{proof}

Finally, we discuss the IET distribution of all links. 
The distribution is represented by a random variable $E$, 
\begin{equation}
	\prob(E \le x) = \frac{1}{|\E|}\sum_{i \in \E} \prob(E_i \le x),
\end{equation}
where $E_i$ is a random variable representing the IET distribution of link $i$, $|\E|$ is the number of links.
By Eq.~(S29), we have
\begin{equation*}
	\min_{i \in \E} \prob(E_i \le x) \le  \prob(E \le x)\le \min_{i \in \E} \prob(E_i \le x),
\end{equation*}
which means the distribution of all links is upper and lower bounded by the distribution of individual links.
When the algorithmic IET distribution of every single link is heavy-tailed (exponential), the distribution of all links is also heavy-tailed (exponential).

\section{Connection to intercommunication time}
A node $x$ is said to be communicating at time $t$ when $x$ is active and at least one of its neighbors is active.
In this case, we can count the time interval between two communication events, that is, the intercommunication time (ICT).
We use the distribution of the first communication time to analyze the statistical property of ICTs.

We assume a node $x$ has $k$ links (i.e. $k$ neighbors), and the stochastic process of $x$ and these neighbors is denoted as $\{X_n \}_{n\ge 0}$ and $\{N_n^{(1)} \}_{n\ge 0}$,$...$, $\{N_n^{(k)}\}_{n\ge 0}$, respectively.
The stopping time for $x$ is 
\begin{equation*}
	\tau_{com}^{(x)}=\min\{\tau_1,...,\tau_k\},
\end{equation*}
where $\tau_j$ ($j=1,...,k$) is the first activation time of link $j$.
We would prove the following conclusions:
In bursty (Poisson-like) activity patterns, (a) the ICT distribution of node $x$ is upper and lower bounded by power-law (exponential) distributions, and (b) the ICT distribution converges exponentially to the IET distribution as $k\to \infty$.

We first prove part (a).
By the definition of the stopping times, we have 
\begin{equation*}
	\{\tau_1 \le n \} \subset \{\tau_{com}^{(x)} \le n \} \subset \{\tau_x \le n \},
\end{equation*}
and part (a) is obtained.

For part (b), without loss of generality, we assume all links have the same IET distribution.
We have
\begin{equation*}
\begin{aligned}
	\prob(\tau_{com}^{(x)} > n) &= \prob(\min\{\tau_1,...,\tau_k\} > n) = \prob(\tau_1 > n, ..., \tau_k > n) \\
								&=\sum_{w^{(x,n)} \circ (w^{(1,n)}+...+w^{(k,n)})=\textbf{0}^{(n)}} \prob(X^{(n)}=w^{(x,n)},N^{(1,n)}=w^{(1,n)},...,N^{(k,n)}=w^{(k,n)}) \\
								&= \sum_{w^{(x,n)} = \textbf{0}^{(n)}} + \sum_{\substack{w^{(x,n)} \neq \textbf{0}^{(n)}, \\ w_x^{(n)} \circ (w^{(1,n)}+...+w^{(k,n)})=\textbf{0}^{(n)}}}.
\end{aligned}
\end{equation*}
The first part of the above equation equals $\prob(\tau_x > n)$.
We turn to prove that the second part converges exponentially to $0$ as $k\to \infty$.

For any trajectory $w^{(x,n)} \neq \textbf{0}^{(n)}$, node $x$ is activated at least once during time $1$ to $n$, we assume the last activation time is $m \le n$.
For any tuple $(w^{(x,n)},w^{(1,n)},...,w^{(k,n)})$ satisfies the condition of the second part, we have 
\begin{equation*}
\begin{aligned}
	\prob(X^{(n)}=w^{(x,n)},...,N^{(k,n)}=w^{(k,n)}) &\le \prob(X^{(m)}=w^{(x,m)},...,N^{(k,m)}=w^{(k,m)}) \\
	&= \prob(X^{(m-1)}=w^{(x,m-1)},...,N^{(k,m-1)}=w^{(k,m-1)}) \\
	&\times \prob \left(\substack{X_m=1,N_m^{(1)}=0,...,N_m^{(k)}=0| \\ X^{(m-1)}=w^{(x,m-1)},...,N^{(k,m-1)}=w^{(k,m-1)}} \right) \\
	&= \prob(X^{(m-1)}=w^{(x,m-1)},...,N^{(k,m-1)}=w^{(k,m-1)}) \\
	& \times \prob(X_m=1|X^{(m-1)}=w^{(x,m-1)}) \\
	& \times \prob(N_m^{(1)}=0|X^{(m)}=w^{(x,m)}, N^{(1,m-1)}=w^{(1,m-1)}) \\
	& ... \\
	& \times \prob(N_m^{(k)}=0|X^{(m)}=w^{(x,m)}, N^{(k,m-1)}=w^{(k,m-1)}). \\
	& \le p_x(w^{(x,m-1)},1) \left(\frac{p_x(w^{(x,m-1)},1) - p_z(\textbf{0}^{(m-1)},1)}{p_x(w^{(x,m-1)},1)}\right)^k
\end{aligned}
,
\end{equation*}
where $p_z(\textbf{0}^{(m-1)},1)$ represents the conditional probability of links.
As the number of the tuple $(w^{(x,n)},w^{(1,n)},...,w^{(k,n)})$ satisfying the condition of the second part is finite and 
$$\frac{p_x(w^{(x,m-1)},1) - p_z(\textbf{0}^{(m-1)},1)}{p_x(w^{(x,m-1)},1)} < 1,$$ 
the second part decay exponentially to $0$.

When the average degree of underlying topologies is sufficiently large, the ICT distribution of nodes is almost the same as the IET distribution.

\section{Statistical test in empirical datasets}

In our theoretical model, we assume that the activity of every single node and trunk follows a renewal process and the activity of nodes is somehow conditionally independent (Eq.~(S4)).
Here we use statistical inference to verify whether the above two assumptions hold in empirical datasets.

For the first assumption, we need to demonstrate that the IET samples of single nodes/links come from the same distribution and are independent of each other.
We begin with counting the IET samples of each node/link. 
For each node/link with $n$ samples, we randomly divide them into two sets of size $n/2$.  
We use the Kolmogorov–Smirnov test \cite{kolmogorov1933sulla} to verify whether these two sets are from the same distribution, which is coincided with the null hypotheses.
And for the conditional independence, we use the Spearman rank correlation test \cite{myers2013research}, in which the null hypothesis is that the two sets are uncorrelated.

Figure S6 shows that the null hypotheses are accepted in both the Kolmogorov–Smirnov test and the Spearman rank correlation test, meaning that the activity of nodes and links is a renewal process in empirical datasets.
Furthermore, the distributions of the $p$-value for empirical datasets match with that for the corresponding synthetic temporal networks.
Another important finding is that although we have proved that the activity of branches is not a strict renewal process, most of the branches still pass the tests.
This indicates that the approximation we have done in Section 3 is reasonable.  

For the second assumption, we cannot directly calculate whether the left-hand side of Eq.~(S4) equals the right-hand side through statistics since we only have one trajectory for each node and link in empirical datasets.
We offer an alternative solution to test the conditional independence in Eq.~(S4).
Specifically, for each pair of nodes in the underlying topology, we count the conditional IET samples of one node $x$ when the other node $y$ is active or inactive, which forms two sets of samples.
Then we use the Kolmogorov–Smirnov test over these two sets. 
If the null hypothesis is valid, Eq.~(S4) holds in empirical datasets.
This is because the relation 
	\begin{equation}
	\begin{aligned}
		&\prob(X_{m+k}=1,...,X_{m+1}=0|X_m=1,X^{(n-1)}=x^{(n-1)}, Y^{(n)}=y^{(n)}) \\
	&= \prob(X_{k}=1,...,X_{1}=0)
	\end{aligned}
	\end{equation}
holds for any $m,k$ under Equation (S4).
It is straightforward to check the left-hand side of Eq.~(30) equals
\begin{equation*}
	\frac{\prob(X_{m+k}=1,...,X_{m+1}=0,X_m=1,X^{(n-1)}=x^{(n-1)}, Y^{(n)}=y^{(n)})}{\prob(X_m=1,X^{(n-1)}=x^{(n-1)}, Y^{(n)}=y^{(n)})},
\end{equation*}
and the numerator equals
\begin{equation*}
	\begin{aligned}
		&\sum_{y_i \in \{0,1\},\ i=1,...,k-1} \prob(X_{m+k}=1, Y_{m+k-1}=y_{m+k-1},...,X_{m+1}=0,Y_{m+1}=y_1, \\
		&X_m=1,X^{(n-1)}=x^{(n-1)}, Y^{(n)}=y^{(n)}) \\
		&= \prob(X_{k}=1|X_{k-1}=0,...,X_1=0) \prob(X_{m+k-1}=0,...,X_m=1,X^{(n-1)}=x^{(n-1)}, Y^{(n)}=y^{(n)}) \\
		&= \prob(X_{k}=1,...,X_{1}=0)\prob(X_m=1,X^{(n-1)}=x^{(n-1)}, Y^{(n)}=y^{(n)}).
	\end{aligned}
\end{equation*}
Therefore, the left-hand side equals the right-hand side of Eq.~(S30).
Figure S7 shows the corresponding statistical results.
The distribution of $p$-values indicates that Eq.~(S4) does not hold in empirical datasets.

Nevertheless, we still find weaker conditional independence in empirical datasets.
Specifically, we use the Chi-squared test \cite{pearson1900x} to validate if the following relation holds
\begin{equation}
\begin{aligned}
	&\prob(X_{n+1}=1|X_n=0,...,X_{n-l+2}=0, X_{n-l+1}=1, Y_n=0) \\
	&= \prob(X_{n+1}=1|X_n=0,...,X_{n-l+2}=0, X_{n-l+1}=1, Y_n=1) 
\end{aligned}
\end{equation}
The parameter $l$ represents the IET from the last activation time.
If Eq.~(S4) holds, the null hypotheses for the test should be valid for any possible $l$.
We find that the null hypotheses become more easily accepted when $l$ is larger (fig.~S8).
This means that the state of $y$ matters only when $x$ tries to activate frequently.

\section{Time-varying underlying topologies}

Another pivotal advantage of our algorithm is easily integrated with the evolution of networks.
Algorithm 4 demonstrates a unified framework to construct temporal networks with evolving underlying topologies.

The most common network evolution model is the network growth model, in which new nodes with links sequentially enter a network system and connect to old nodes \cite{barabasi1999} (including the example presented in the main text).
In addition, when considering recessionary effects, the number of nodes and links may decrease \cite{saavedra2008asymmetric}.
For different network evolution, the main design in Algorithm 4 is how to update $\mathcal{T}$.
A reasonable design is necessary when nodes and links may vanish on the underlying topology during the evolution.

\begin{algorithm}[t]
\caption{Combination with network evolution} 
\hspace*{0.02in} {\bf Input:} 
initial underlying topology $\G$ and parameter $t_{tol}$ \\
 \hspace*{0.02in} {\bf Output:} 
trajectories of all nodes and links
\begin{algorithmic}[1]
\State Select a spanning tree $\mathcal{T}$ of $\G$
\State Assign a probability mass function to each node and trunk in $\mathcal{T}$
\For{$t=1$ to $t_{tol}$} 
	\State Execute a single loop of Algorithm 3 on $\G$
	\If{$\G$ evolves}
		\State Update $\G$ by network evolution
		\State Update $\mathcal{T}$ on the new underlying topology
		\State Assign a targeted distribution to each new node/trunk and set them to be active
		\State Update the state of old nodes and links
	\EndIf
\EndFor
\Return  trajectories of all nodes and links
\end{algorithmic}
\end{algorithm}

\section{Statistics for measuring burstiness and temporal correlations}
The burstiness parameter $B$ is widely used to measure the level of burstiness \cite{goh2008}, defined by the coefficient of variation,
\begin{equation*}
	B = \frac{\sigma/\mu-1}{\sigma/\mu+1},
\end{equation*}
where $\mu$ and $\sigma$ are the mean and standard deviation of IET distributions.
When $\mu$ and $\sigma$ are finite, the definition is meaningful and $|B| < 1$. 

We calculate the burstiness parameter of nodes and links for the synthetic temporal networks in Fig.~2. 
In order to compare algorithmic and theoretical burstiness parameters, we set a cutoff $\kappa$ for theoretical calculation. 
The theoretical mean $\mu$ and standard deviation $\sigma$ for a power-law distribution are
\begin{equation*}
	\mu = \frac{\sum_{i=1}^{\kappa} i^{-\gamma+1}}{\sum_{i=1}^{\kappa} i^{-\gamma}}, \quad \sigma = \left(\frac{\sum_{i=1}^{\kappa} i^{-\gamma+2}}{\sum_{i=1}^{\kappa} i^{-\gamma}} - \mu^2 \right)^{\frac{1}{2}},
\end{equation*} 
and for a discrete exponential distribution are 
\begin{equation*}
	\mu = e^{\gamma/2}\sum_{i=1}^{\kappa} i(e^{-\gamma(i-0.5)}-e^{-\gamma(i+0.5)}), \quad \sigma = \left(e^{\gamma/2}\sum_{i=1}^{\kappa} i^2(e^{-\gamma(i-0.5)}-e^{-\gamma(i+0.5)}) - \mu^2 \right)^{\frac{1}{2}}.
\end{equation*}
Tables S1 and S2 compare the simulation and theoretical $B$ in bursty activity patterns and in Poisson-like activity patterns, respectively.
The simulations are robust to underlying topologies and are well-predicted by theoretical results.
Both nodes and links show a high level of burstiness in bursty activity patterns and present a negative (low) level in Poisson-like activity patterns.

We also investigate the temporal correlation of synthetic temporal networks.
The autocorrelation function $A(\Delta t)$ is a common statistic to appreciate the global activity correlations for temporal networks.
For a temporal network $\G=\{G_1,...,G_T\}$,
the autocorrelation coefficient $A(\Delta t)$ is defined as
\begin{equation*}
	A(\Delta t) = \frac{\frac{1}{T-\Delta t}\sum_{i=1}^{T-\Delta t}E(i)E(i+\Delta t)-\mu_1 \mu_2}{\sigma_1 \sigma_2},
\end{equation*}
where $E(i)$ denotes the total activation numbers of nodes or links in the snapshot $i$, $\mu_1, \sigma_1^2$ (respectively $\mu_2, \sigma_2^2$) are the sample mean and sample variance of the series $\{E(i)\}_{i=1}^{T-\Delta t}$ (respectively $\{E(i+\Delta t)\}_{i=1}^{T-\Delta t}$).
The parameter $\Delta t$ represents the distance of the time windows in the two series.
In particular, when $\Delta t=1$, $A(\Delta t)$ is called the memory coefficient \cite{goh2008}.

By Hölder's inequality, we obtain that $|A(\Delta t)|\le 1$. 
The closer $A(\Delta t)$ is to $0$, the less correlated $\{E(i)\}_{i=1}^{T-\Delta t}$ and $\{E(i+\Delta t)\}_{i=1}^{T-\Delta t}$, and the weaker the autocorrelation of the temporal network.
When $A(\Delta t)$ is close to 1 or -1, A strong positive or negative linear correlation exists between series $\{E(i)-\mu_1 \}_{i=1}^{T-\Delta t}$ and series $\{E(i+\Delta t)-\mu_2 \}_{i=1}^{T-\Delta t}$.

Figure S10 shows the results of $A(\Delta t)$ in bursty activity patterns and Poisson-like activity patterns.
In bursty activity patterns, the construction process is a non-Markovian process, and the results of $A(\Delta t)$  show positive temporal correlations for all time intervals $\Delta t$, which is consistent with heterogeneous temporal behaviour discovered in empirical temporal networks.
In Poisson-like activity patterns, the results of $A(\Delta t)$ are almost $0$.
Actually, Eq.~(S10) indicates that the activity of single nodes/trunks is a discrete-time Markov chain $\{M_n \}_{n\ge 0}$ with two states $\{s_0, s_1\}$, where $s_0$ represents the element (node or trunk) is inactive and $s_1$ represents the element is active.
We set $s_0=0$ and $s_1=1$. 
The corresponding transition probability matrix is given by follows
\begin{spacing}{1.5}
	\centerline{$\bordermatrix{%
& s_0 & s_1 \cr
s_0 & e^{-\lambda} & 1-e^{-\lambda} \cr
s_1 & e^{-\lambda} & 1-e^{-\lambda} \cr
}$,}
\end{spacing} 
\noindent
where $\lambda$ is the exponent of the targeted distribution.
For all $m > 1$, we have
\begin{equation*}
	\begin{aligned}
		\prob(M_{m}=1) &= \sum_{i_1,...,i_{m-1} \in \{0,1\}} \prob(M_1=i_1,...,M_{m-1}=i_{m-1},M_{m}=1) \\
					   &= \sum_{i_1,...,i_{m-1} \in \{0,1\}} \prob(M_1=i_1,...,M_{m-1}=i_{m-1})\prob(M_m=1|M_{m-1}=i_{m-1}) \\
					   &= 1-e^{-\lambda}.
	\end{aligned}
\end{equation*}
Hence, for all $n,m \ge 1$, we have
\begin{equation*}
\begin{aligned}
	\prob(M_n=1, M_{n+m}=1) &= \sum_{i_1,...,i_{m-1} \in \{0,1\}} \prob(M_n=1, M_{n+1}=i_1,...,M_{n+m-1}=i_{m-1},M_{n+m}=1) \\ 
							&= \sum_{i_1,...,i_{m-1} \in \{0,1\}} \prob(M_n=1,...,M_{n+m-1}=i_{m-1})\prob(M_{n+m}=1|M_{n+m-1}=i_{m-1}) \\
							&= (1-e^{-\lambda})^2. \\
\end{aligned}
\end{equation*}
This gives
\begin{equation*}
\begin{aligned}
	A(\Delta t) &\sim \frac{1}{T-\Delta t}\sum_{i=1}^{T-\Delta t}\mathbb{E} M_i M_{i+\Delta t} - \frac{1}{(T-\Delta t)^2}\sum_{i,j=1}^{T-\Delta t} \mathbb{E}M_i \mathbb{E}M_{i+\Delta t} \\
	&= \frac{1}{T-\Delta t}\sum_{i=1}^{T-\Delta t} \prob(M_i=1, M_{i+\Delta t}=1) - \frac{1}{(T-\Delta t)^2}\sum_{i,j=1}^{T-\Delta t} \prob(M_i=1) \prob(M_{i+\Delta t}=1)\\
	&= 0,
\end{aligned}
\end{equation*}
meaning that Poisson-like activity patterns are memoryless (or homogenous).

\section{Analysis of aggregated networks}
In addition to studying the IET distribution, we analyze the structural measures of aggregated networks (fig.~S11) with different aggregation times $t_{agg}$ to evaluate other temporal properties of our synthetic temporal networks. 
We study two typical measures, the expected number of activations of individual nodes (links) and the node strength distribution of aggregated networks. 
The former is a micro statistic capturing the frequency of activations of each unit in networks, and the latter is a macro statistic showing the structural information of entire networks.

For a node $i$ (respectively trunk $j$) whose targeted IET distribution is $\phi_i(\Delta t, \alpha_i)$ (respectively $\psi_j(\Delta t, \beta_j)$), the expectation of $\phi_i$ (respectively $\psi_j$) is denoted as $\mu(\alpha_i)$ (respectively $v(\beta_j)$), where the exponent $\alpha_i$ (respectively $\beta_j$) is obtained by sampling from a distribution $\eta_{node}$ (respectively $\eta_{link}$).
The total activation number of node $i$ up to moment $t$ is denoted as $A_t^{(i)}$.
Let
\begin{equation*}
	A^{(i)}(t) = \mathbb{E}A_t^{(i)}.
\end{equation*}
Using the elementary renewal theorem \cite{ross1996}, we have 
\begin{equation}
	\frac{A^{(i)}(t)}{t}\to \frac{1}{\mu(\alpha_i)} \quad \text{as}~t \to \infty,
\end{equation}
where $\frac{1}{\infty}=0$. 
The conclusion for single links is similar.
Equation (S32) presents an intuitive conclusion that the average growth rate of activation numbers asymptotically equals the frequency of activations.
In a bursty activity pattern, when $\alpha < 2$, the expectation $\mu(\alpha) = \infty$.
This suggests that the growth is sublinear, and the rate asymptotically equals 0.
In a Poisson-like activity pattern, for any exponent $\alpha$, $\mu(\alpha) < \infty$, thus the growth is linear.

Another important statistic for an aggregated network is the node strength distribution. 
In empirical datasets, we find that the node strength distributions present specific robustness across different time scales (fig.~S12).
Here we prove that our model also reproduces this property.

For a static unweight network $\G$, the degree distribution of $\G$ is denoted as $d(x)$ and its maximum value is $k_{max}$.
We set a random variable $X$ of which the probability mass function is $d(x)$.
Let $N_t$ denote the strength of a node in the aggregated network generated by $\G$ with the aggregation time, $t_{agg}=t$, $N_t$ is a random variable. 
We assume that the exponent of all links is sampled from $\eta_{link}$, using a limit theorem of renewal theory \cite{ross1996}, with probability $1$,
\begin{equation}
	\frac{N_t}{t} \to \sum_{i=1}^{X} \frac{1}{v(\eta_i)} \quad \text{as} \ t\to \infty,
\end{equation}
where $\{\eta_i\}_{1\le i \le k_{max}}$ is a sequence of independent random variables with a common distribution $\eta_{link}$.

As $X$ and $\{\eta_i\}_{1\le i \le k_{max}}$ are independent, when $t$ is sufficiently large, from Eq.~(S33), the distribution of $N_t$ is given by
\begin{equation}
	\prob(N_t\le s)=\sum_{x=1}^{k_{max}}\prob(X=x)F^{(x)}(\frac{s}{t}),
\end{equation}
where $F^{(x)}$ is the x-order convolution of the distribution function of the random variable $1/v(\eta_1)$.
The Laplace transform of the random variable $X$ is defined by $\rho_{X}(s) = \mathbb{E} e^{-sX}$,
then Eq.~(S34) is converted into 
\begin{equation*}
	\rho_{N_t}(s) = \sum_{x=1}^{k_{max}} \prob(X=x) [\rho_{\frac{1}{v(\eta_1)}}(st)]^x.
\end{equation*}
In particular, when a.s. $\eta_i$ is a constant and equals $\beta$, Eq.~(S34) can be estimated as follows
\begin{equation}
	\prob(N_t \le s) \approx \int_{0}^{s v(\beta)/t}d(x) \text{d}x.
\end{equation}  
Formally, let $p_{N_t}(x)$ denote the probability density function of $N_t$. 
From Eq.~(S35), we have 
\begin{equation}
	p_{N_t}(x) = \frac{v(\beta) d(\frac{x v(\beta)}{t})}{t}.
\end{equation} 

Equation (S36) shows the relationship between the degree distribution of a static network and the node strength distribution of the aggregated network.
When $d(x)=C x^{-\gamma}$, the node strength distribution is also power-law with the same exponent $\gamma$ for each aggregation time, indicating the robustness of aggregated networks.
Figure S13A shows the survival function of node strength with different $t_{agg}$ based on scale-free underlying topologies. 
As one can see, the results are all power-law distributions with the same exponent as the degree distribution, which suggests that the distribution of node strength is robust to the aggregation time and the given distribution.

For general degree distributions, we can also obtain a similar robust behaviour of node strength distributions.
For two aggregated networks $G_1, G_2$ with the aggregation time $t_1, t_2$ ($t_1 > t_2$), we have 
\begin{equation*}
	\prob(X t_2 / v(\beta) > s) = F^{(2)}(s) = F^{(1)}(s\frac{t_1}{t_2}):=F^{(1)}(\hat{s}),
\end{equation*}
where $F^{(i)}(s)$ is the survivor function of the random variable $X t_i / v(\beta)$ ($i=1,2$).
When $s$ in $F^{(2)}$ changes from $s_1$ to $s_1+1$, $\hat{s}$ in $F^{(1)}$ changes from $s_1\frac{t_1}{t_2}$ to $(s_1+1)\frac{t_1}{t_2}$, which means the proportion of nodes with strength between $s_1\frac{t_1}{t_2}$ and $(s_1+1)\frac{t_1}{t_2}$ in $G_1$ is same as the proportion of nodes with strength $s_1$ in $G_2$.

The normalization is executed as follows.
We first select a sufficiently large moment, $t_{base}$, and its corresponding aggregated network, $G_{base}$, is said to be  the baseline.
For any aggregation time $t_{agg} \ge t_{base}$, when all links have the same targeted distribution, the proportion of nodes with strength between $st_{agg}/t_{base}$ and $(s+1)t_{agg}/t_{base}$ in $G_{agg}$ is same as that with strength $s$ in $G_{base}$.
Therefore, the normalized distribution of node strength for the aggregated network $G_{t_{agg}}$ is the same as that for $G_{t_{base}}$. 
We verify the above robustness on the small-world underlying topology under different activity patterns (fig.~S13B).
The normalized distribution for different aggregation times all collapses onto the node strength of the baseline network.

\clearpage
\begin{table}[h]
\centering
\caption{Burstiness parameter of bursty activity patterns. The cutoff $\kappa$ is set to be $1\times 10^4$. The exponent of the simulation is the fitted exponent of algorithmic distributions, and the exponent of the theory is set to the average of the two simulation exponents.}
\begin{tabular}{lcccccc}
Dataset & Object & Exponent & $B$ & Object & Exponent & $B$ \\
\midrule
Simulation (BA) & Nodes & 1.80 & 0.82 & Links & 1.28 & 0.64\\
Simulation (SW) & Nodes & 1.80 & 0.82 & Links & 1.28 & 0.64\\
Theory & Nodes & 1.80 & 0.84 & Links & 1.28 & 0.60\\
Simulation (BA) & Nodes & 2.00 & 0.83 & Links & 1.69 & 0.80\\
Simulation (SW) & Nodes & 2.00 & 0.83 & Links & 1.71 & 0.79\\
Theory & Nodes & 2.00 & 0.86 & Links & 1.70 & 0.81\\
\bottomrule
\end{tabular}
\end{table}

\clearpage
\begin{table}[h]
\centering
\caption{Burstiness parameter of Poisson-like activity patterns. The cutoff $\kappa$ is set to be $10^3$. The implication of parameters is the same as Table 1.}
\begin{tabular}{lcccccc}
Dataset & Object & Exponent & $B$ & Object & Exponent & $B$ \\
\midrule
Simulation (BA) & Nodes & 1.80 & -0.42 & Links & 1.21 & -0.29\\
Simulation (SW) & Nodes & 1.80 & -0.42 & Links & 1.22 & -0.29\\
Theory & Nodes & 1.80 & -0.42 & Links & 1.21 & -0.29\\
Simulation (BA) & Nodes & 2.50 & -0.55 & Links & 1.86 & -0.43\\
Simulation (SW) & Nodes & 2.50 & -0.55 & Links & 1.89 & -0.44\\
Theory & Nodes & 2.50 & -0.55 & Links & 1.87 & -0.43\\
\bottomrule
\end{tabular}
\end{table}

\clearpage
\begin{figure}
\includegraphics[width=\textwidth]{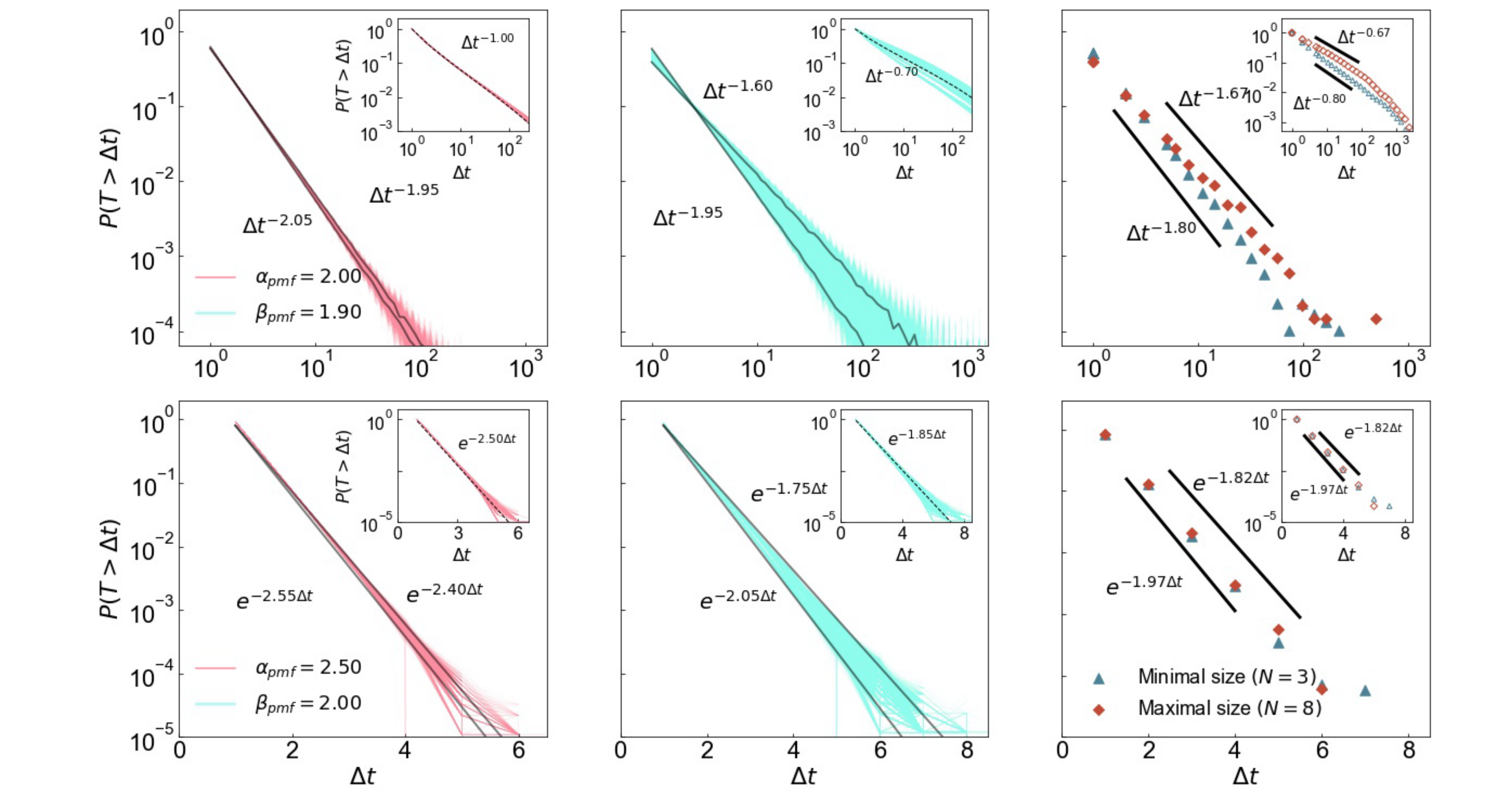}
\caption{\textbf{Single IET distributions of every element.} 
The algorithmic IET distribution of every single node (link) is represented by a red (green) line, respectively, showing in the first and second columns.
The black lines are the upper or lower bounds of the algorithmic results, which are power-law distributions in the first row and exponential distributions in the second row generated by simulations.
We also select the results of branches located in the largest and smallest ring showing in the third column by diamonds and triangles.
The branch located in the largest ring has a smaller fitted exponent than that in the smallest ring, which is consistent with the theoretical result.
}
\end{figure}

\clearpage
\begin{figure}[p]
\includegraphics[width=\textwidth]{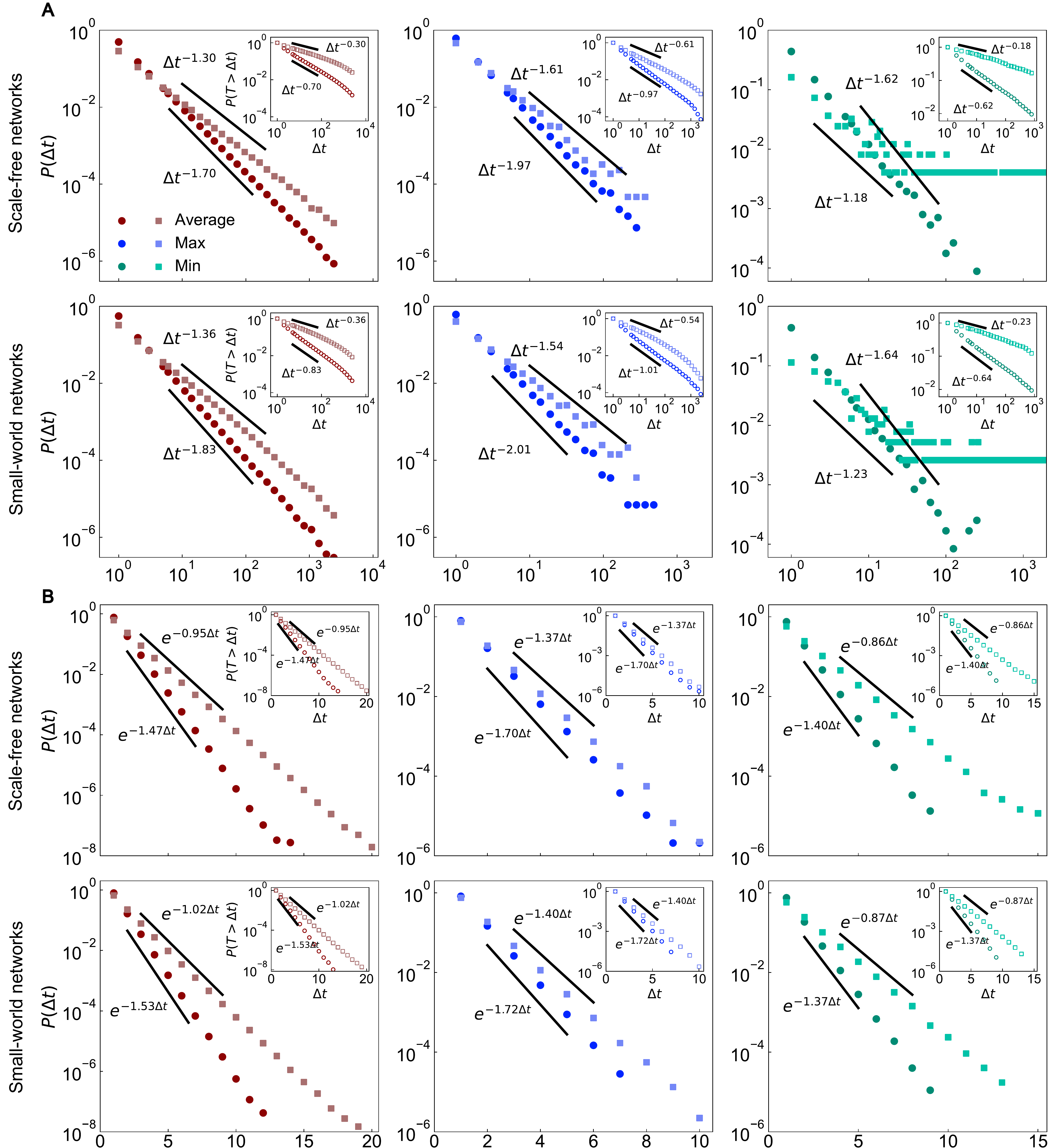}
\caption{\textbf{IET distributions when exponents are sampled from distributions.} 
We consider the case that the exponent of every single node (trunk) is a random variable $\eta_{node}$ ($\eta_{link}$). 
We set $\eta_{node}$ and $\eta_{link}$ to a uniform $[1.60, 2.00]$ random variable and a uniform $[1.20, 1.50]$ random variable in the bursty activity pattern (\textbf{A}) and a uniform $[1.40, 1.70]$ random variable and a uniform $[1.05, 1.30]$ random variable in the Poisson-like activity pattern (\textbf{B}).
The results of nodes (links) are represented by circles (squares).
The aggregated distribution of nodes/links is plotted by brown markers, and the single distribution of the node/link with the maximal (minimal) activation numbers is plotted by blue (green) markers.
Parameter values are the same as those in Fig.~2.}
\end{figure}

\clearpage
\begin{figure}[t]
\includegraphics[width=1\textwidth]{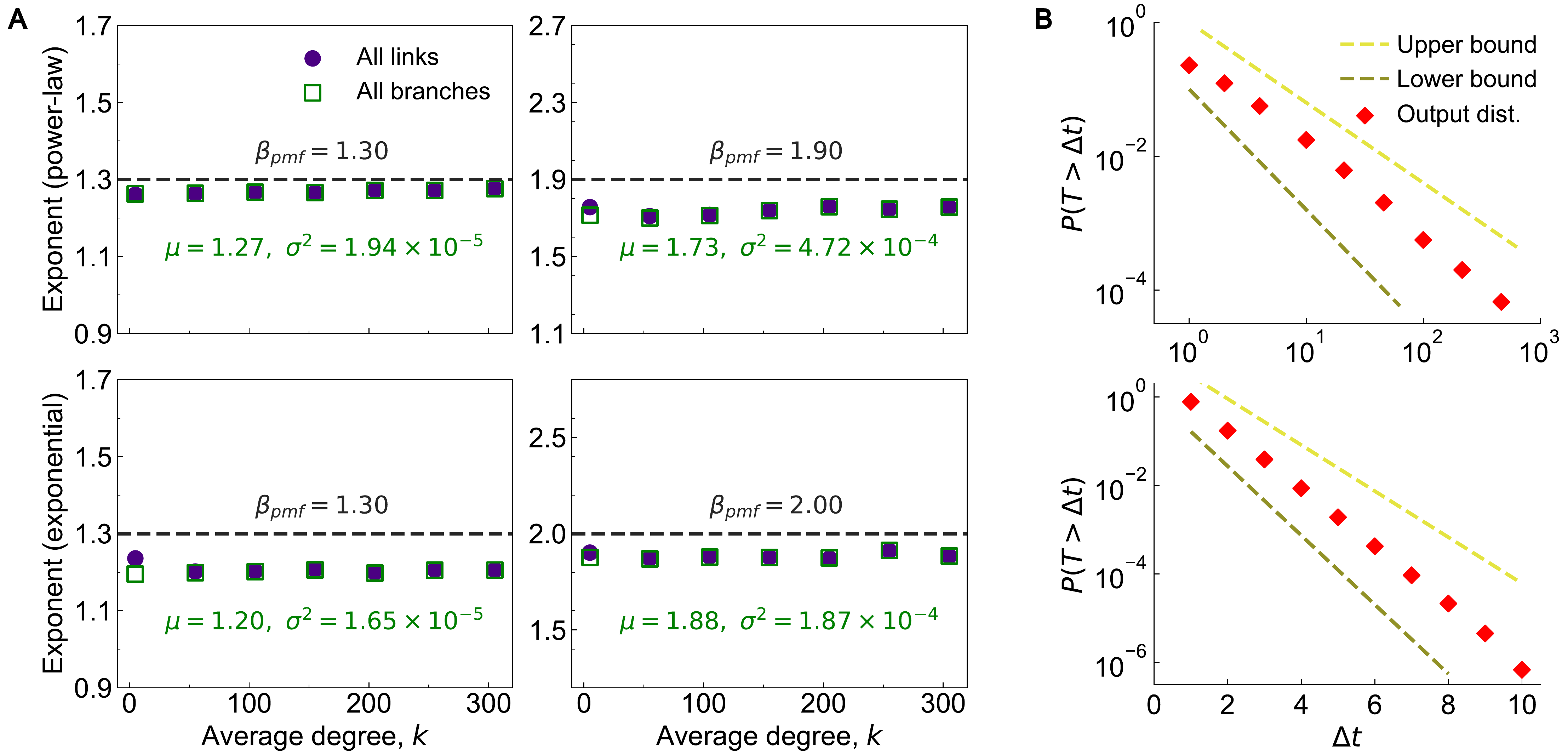}
\centering
\caption{\textbf{Aggregated IET distribution of links is robust to underlying topologies.}
(\textbf{A}) Relationship between the underlying topology and the algorithmic exponent.
Random regular graphs with different average degrees, $k$, are utilized to represent different underlying topologies.
We plot the algorithmic exponent for all branches (hollow squares) and all links (i.e. all trunks and branches, solid circles).  
Black dashed lines show the targeted exponent for trunks.
When the average degree becomes larger, the proportion of branches in all links increases. As a result, the algorithmic exponent for all links gradually converges to that for all branches. 
Furthermore, the algorithmic exponents do not correlate significantly with $k$. The largest magnitude of variances is no more than $10^{-3}$, indicating the robustness to the underlying topology.   
Parameter value: the size of underlying topologies $N=400$.
(\textbf{B}) Schematic illustration of the upper and lower bounds for individual branches. The algorithmic distribution of every single branch (red diamonds) has uniform upper and lower bounds, which are heavy-tailed distributions in a bursty activity pattern (top panel) and exponential distributions in a Poisson-like activity pattern (bottom panel).
}
\end{figure}

\clearpage
\begin{figure}
\includegraphics[width=1\textwidth]{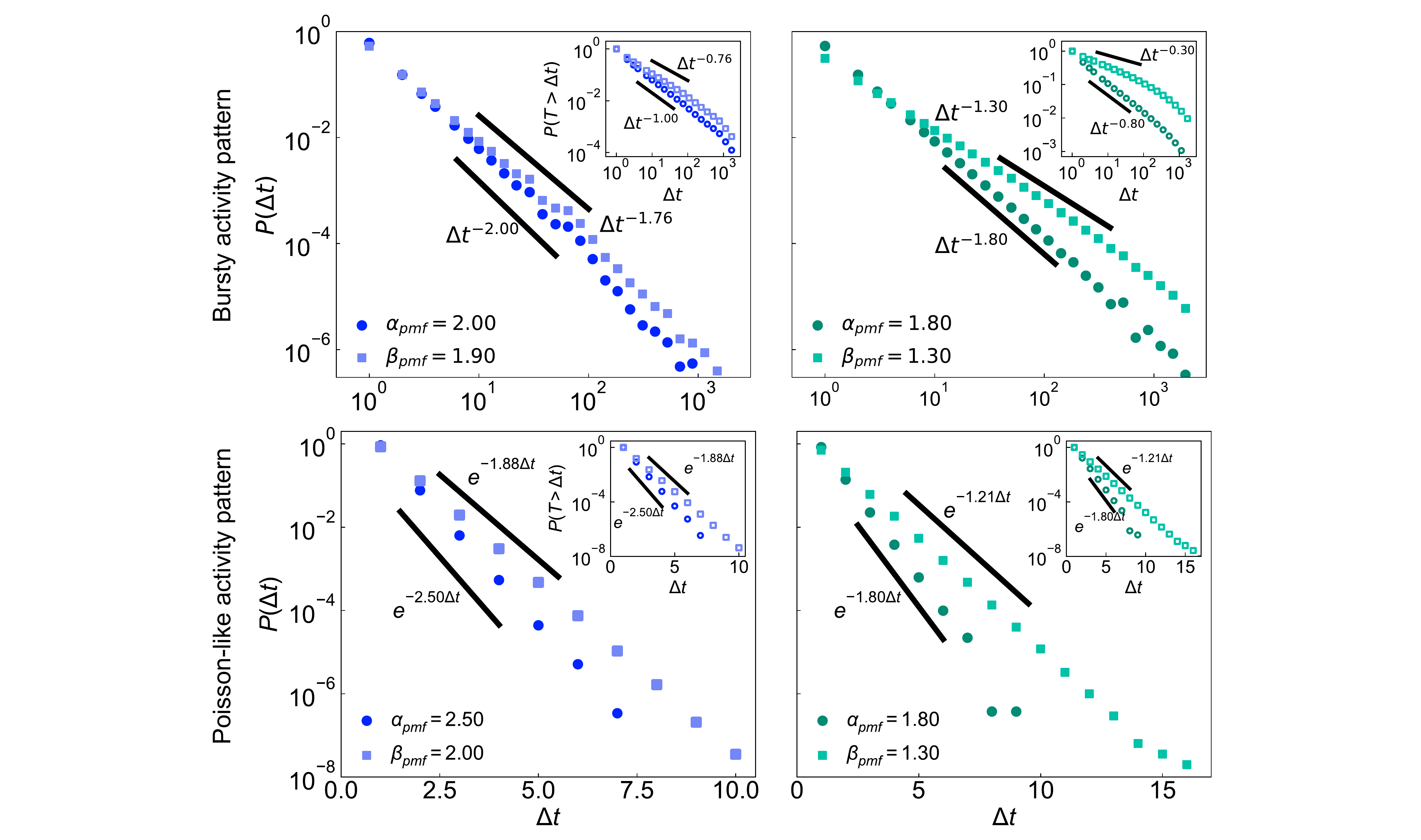}
\centering
\caption{\textbf{IET distributions on well-mixed networks.} 
The largest relative deviation of the algorithmic exponent between Fig.~2 and this is $2\times |1.76-1.69|/(1.76+1.69) \approx 4 \%$.
Parameter value: the size of underlying topologies $N=400$.
Other parameter values are the same as in Fig.~2.}
\end{figure}

\clearpage
\begin{figure}[p]
\includegraphics[width=\textwidth]{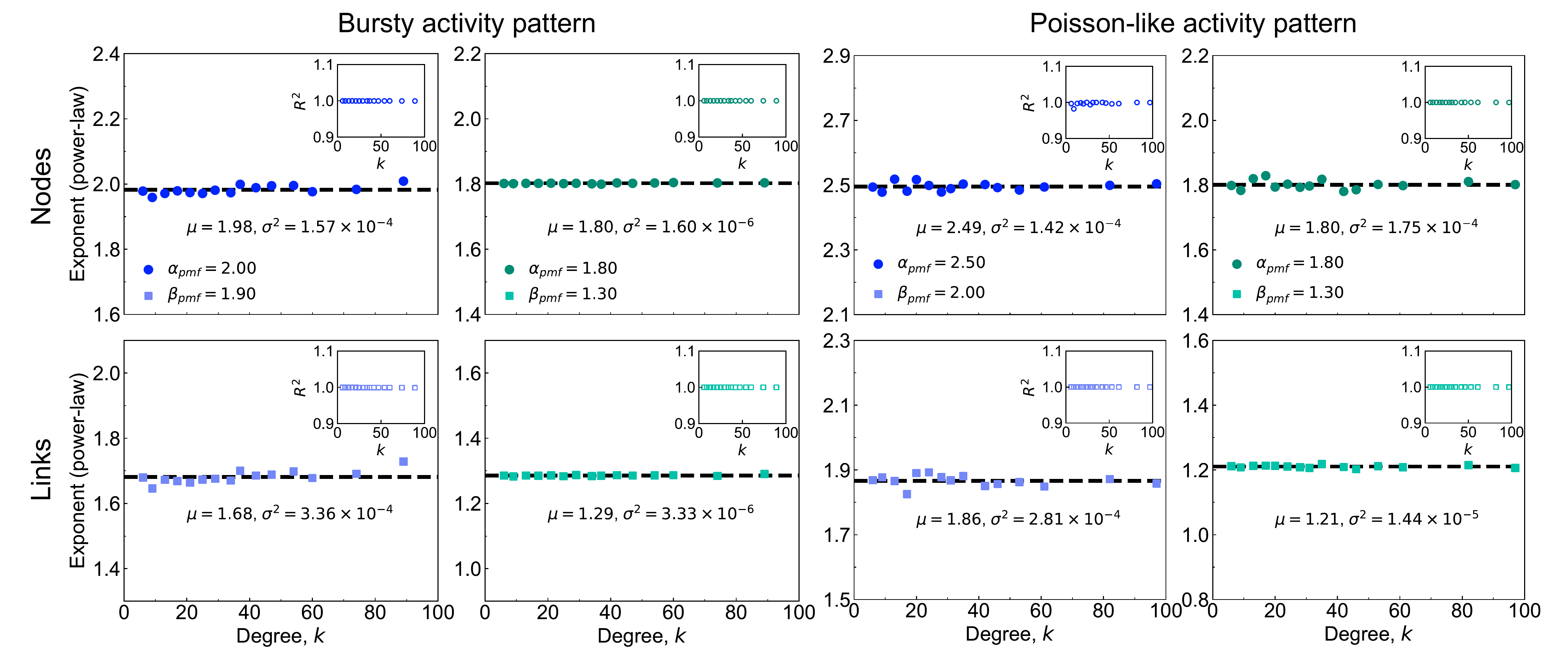}
\caption{\textbf{Aggregted IET distribution of links is robust to the selection of spanning trees.} 
The underlying topology is the Barabási-Albert scale-free network. 
Nodes with different degrees are selected to be the root of a spanning tree to represent the difference in spanning tree selection. 
The algorithmic fitted exponent does not correlate significantly with the selection of spanning trees under both patterns, and the largest magnitude of variances is no more than $10^{-4}$, indicating the robustness to spanning tree selection.
}
\end{figure}

\clearpage
\begin{figure}[p]
\includegraphics[width=1\textwidth]{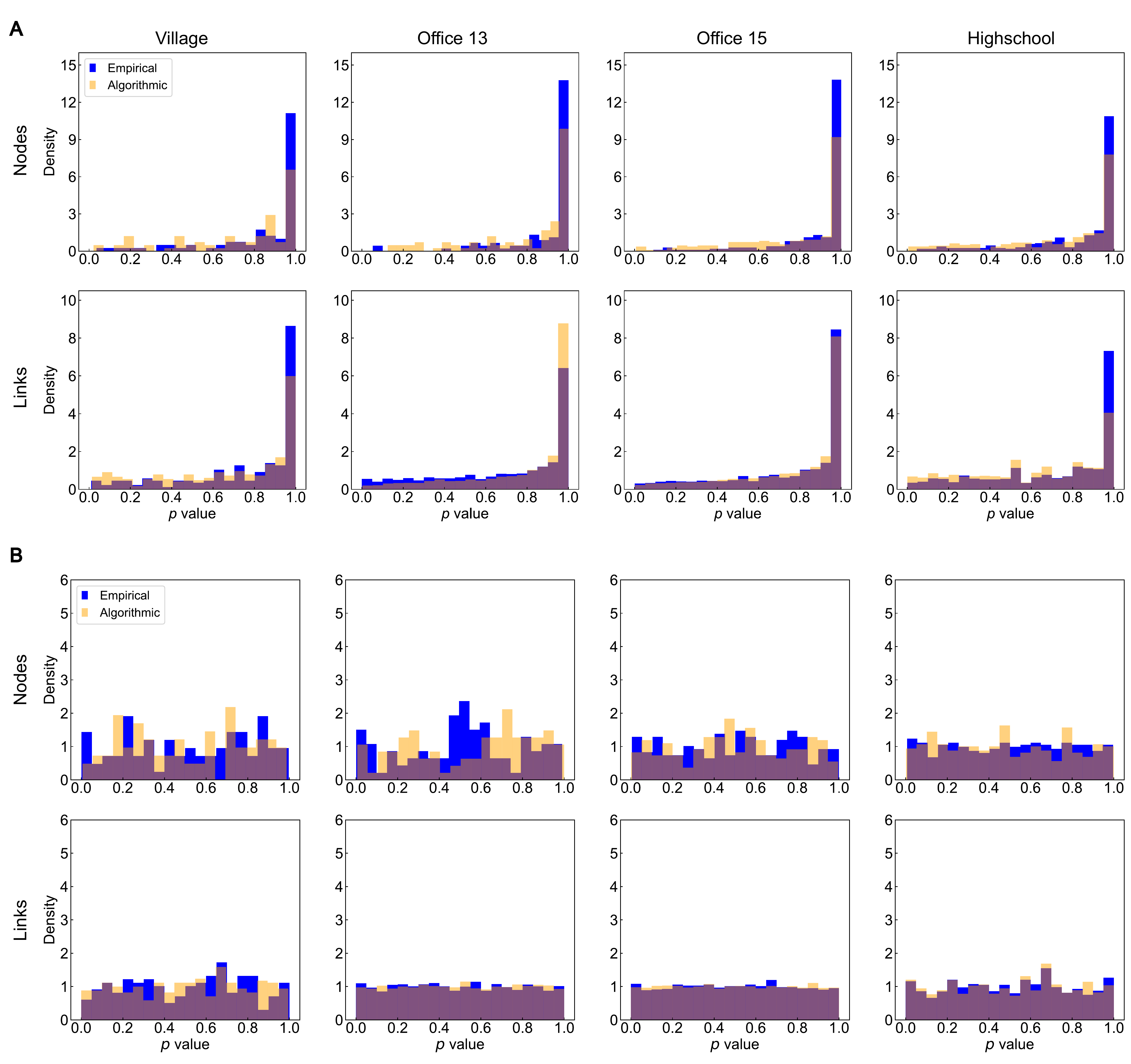}
\centering
\caption{\textbf{Empirical evidence of the activity of nodes/links being renewal processes.} 
We consider the statistical inference for the assumptions that the IET samples of each node/edge come from the same distribution (\textbf{A}) and are independent of each other (\textbf{B}).
For each node/link, we first obtain all its IET samples and randomly divided them into two sets.
For (\textbf{A}), we use the K-S test to judge whether these two sets are from the same distribution.
For (\textbf{B}), we use the Spearman rank correlation test to judge whether these two sets are independent.
We obtain the $p$-values of nodes and links with more than 6 IET samples and display them as histograms.
Blue (yellow) bars show the distribution of the $p$-value for empirical (synthetic) datasets.
}
\end{figure}

\clearpage
\begin{figure}[p]
\includegraphics[width=1\textwidth]{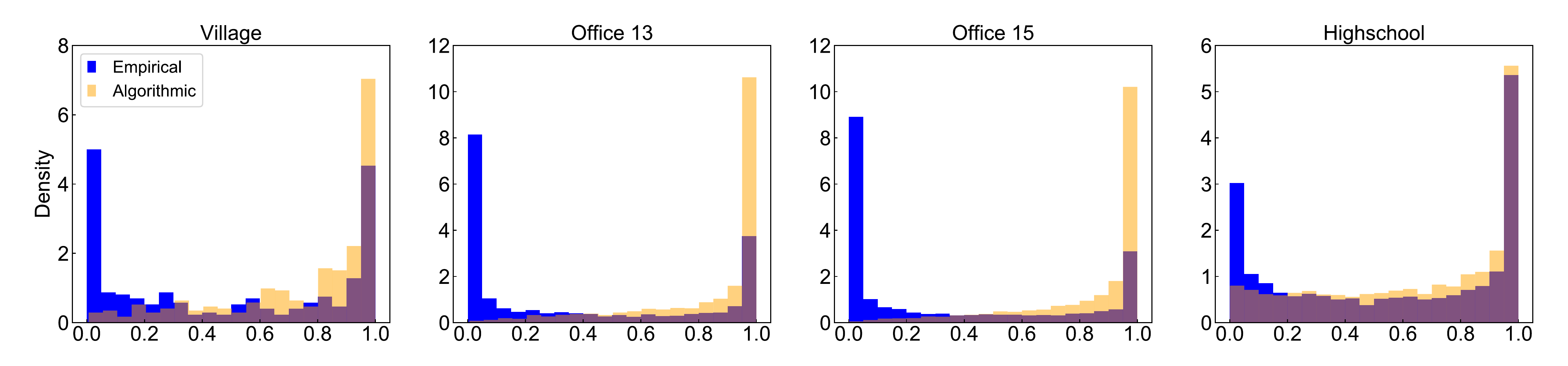}
\centering
\caption{\textbf{Testing Eq.~(S30) in empirical datasets.} 
For each pair of nodes, we use the K-S test to judge whether the IET distribution of one node remains identical when the current state of the other node is active or inactive. 
The synthetic temporal networks definitely can pass the test, but the empirical datasets reject the conditional independence in Eq.~(S30).
}
\end{figure}

\clearpage
\begin{figure}[p]
\includegraphics[width=1\textwidth]{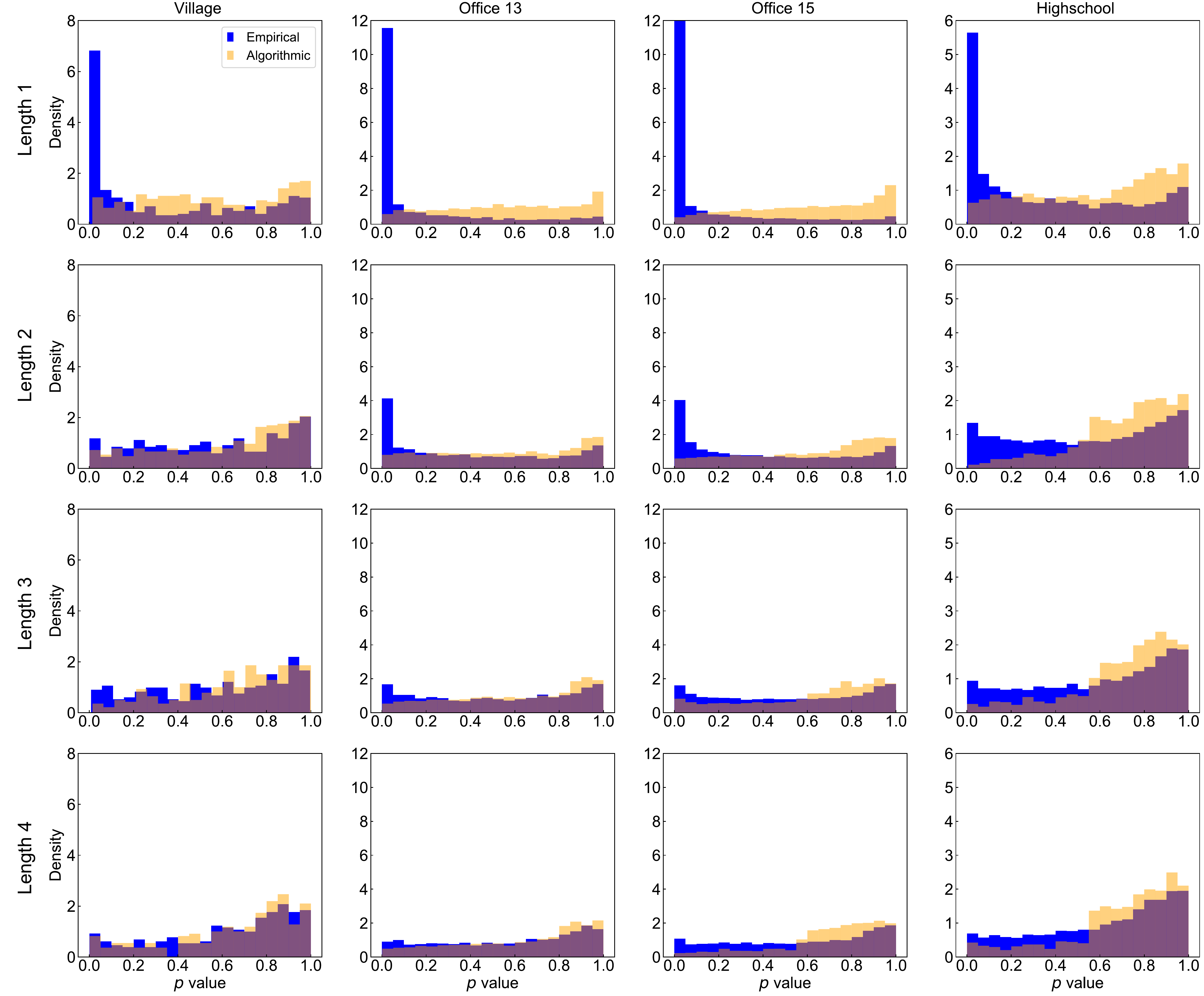}
\centering
\caption{\textbf{Conditional independence tests under different memory lengths.} 
We test Eq.~(S31) with the Chi-squared test under different lengths $l$.
When $l$ becomes larger, the distributions of $p-$values in the empirical datasets is more closed to that in the corresponding synthetic temporal networks.
Therefore, the conditional independence under a large length (such as $l=3,4$) is valid.
}
\end{figure}

\clearpage
\begin{figure}[p]
\includegraphics[width=\textwidth]{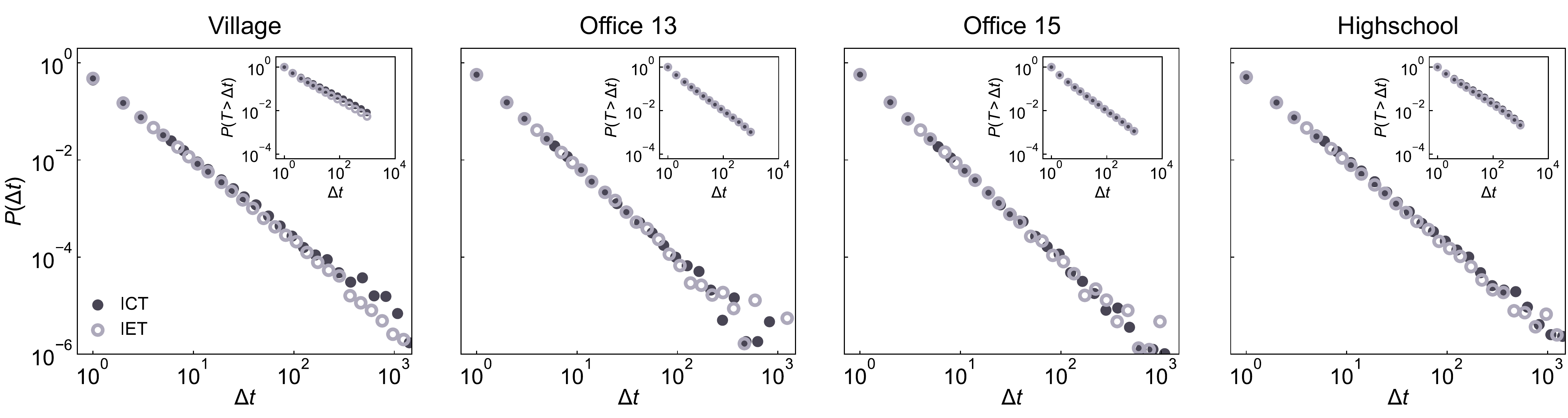}
\caption{\textbf{Comparison between the IET and ICT distributions of nodes.} 
The average degree of the underlying topologies for these datasets is 8.24, 82.42, 152.74, and 35.84, respectively, showing a high level of connectivity among populations.
Therefore, the ICT distributions collapse into the IET distributions.
}
\end{figure}

\clearpage
\begin{figure}[p]
\includegraphics[width=1\textwidth]{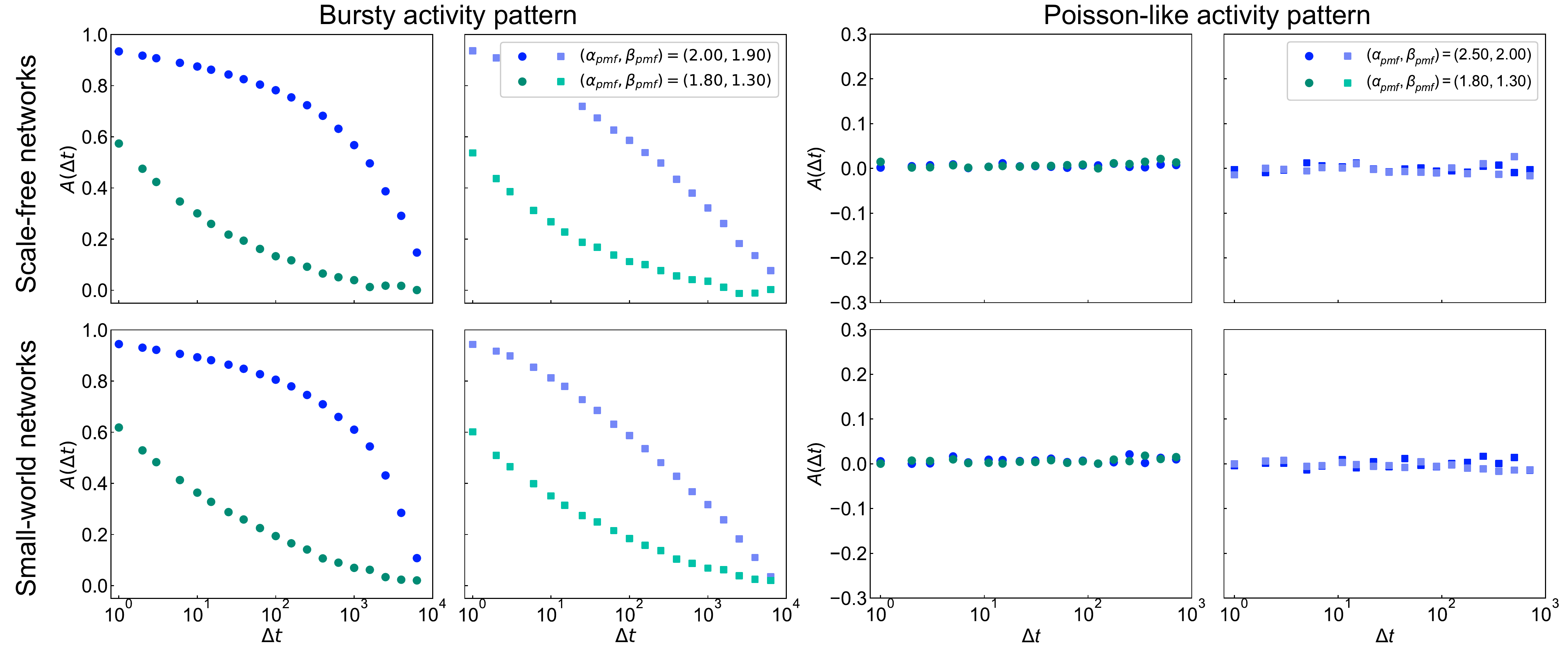}
\caption{\textbf{Temporal correlation of temporal networks.} 
The relationship between the autocorrelation function, $A(\Delta t)$, and the time interval, $\Delta t$, are plotted based on bursty and Poisson-like activity patterns.
$A(\Delta t)$ is always positive for any $\Delta t$ in bursty activity patterns, while $A(\Delta t)$ is approximately 0 in Poisson-like activity patterns.
}
\end{figure} 

\clearpage
\begin{figure}[p]
\includegraphics[width=1\textwidth]{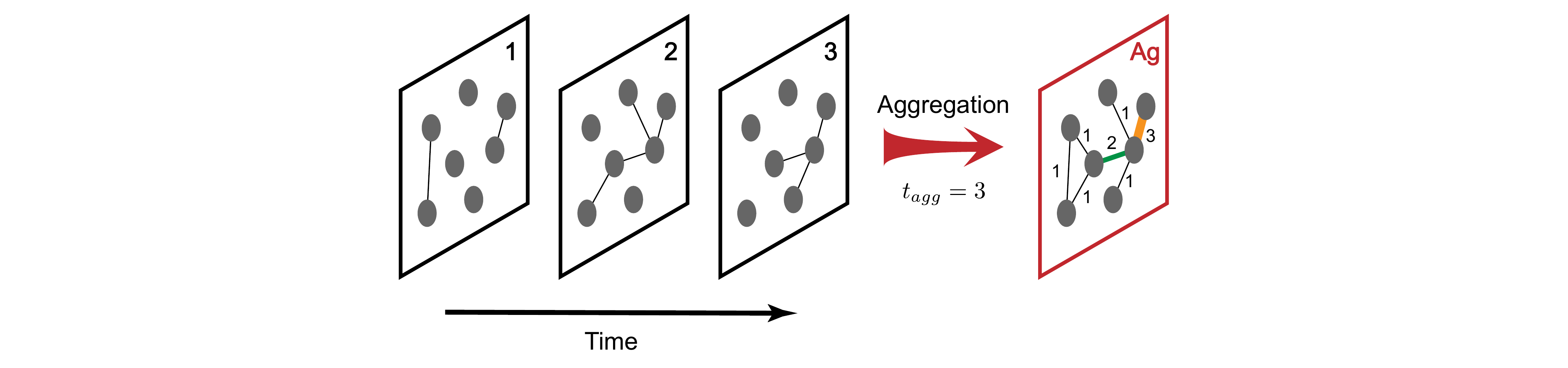}
\centering
\caption{\textbf{Illustration of network aggregation.} The aggregated network (network Ag) collects all interactions from the first $t_{agg}$ snapshots.
The weight of links in the aggregated network represents the activation numbers during the network evolution, which may be larger than $1$ (such as the orange and green links).
}
\end{figure}

\clearpage
\begin{figure}[p]
\includegraphics[width=1\textwidth]{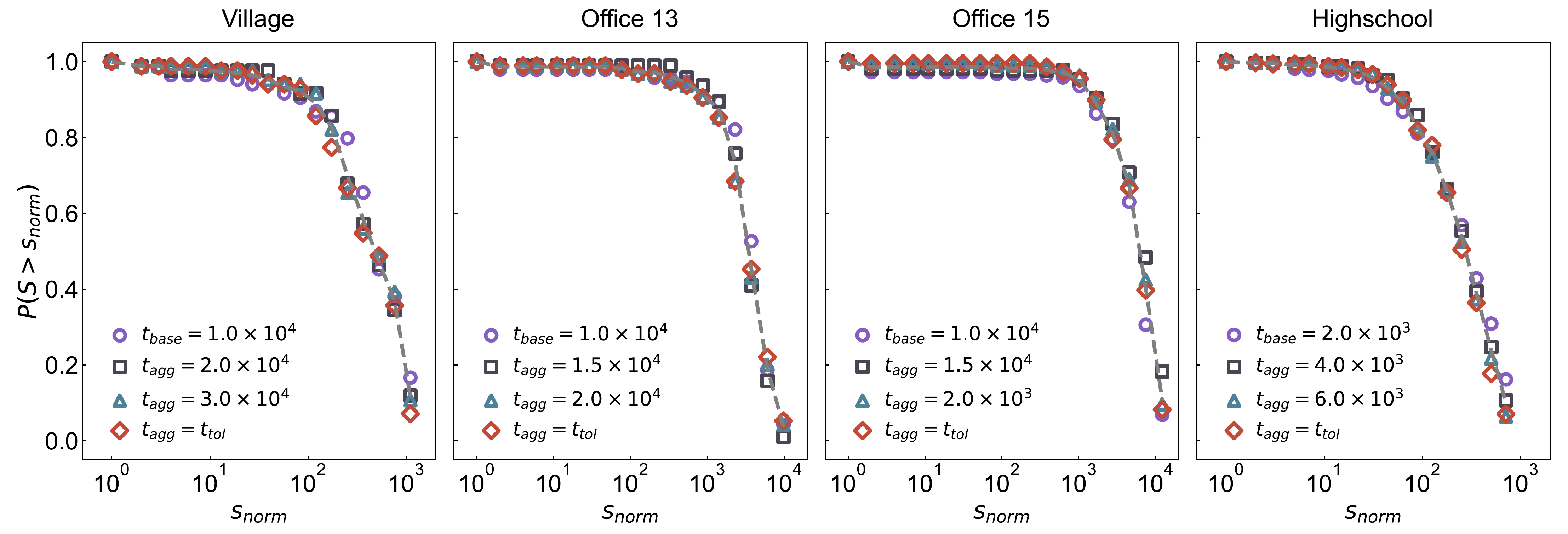}
\centering
\caption{\textbf{Robustness of node strength distributions on empirical temporal networks.} 
The normalized distribution of node strength for an aggregation time $t_{agg} > t_{base}$ collapses onto the node strength distribution of the baseline network with the aggregation time $t_{base}$. 
The grey dashed lines show the average of the distributions.
}
\end{figure}

\clearpage
\begin{figure}[p]
\includegraphics[width=1\textwidth]{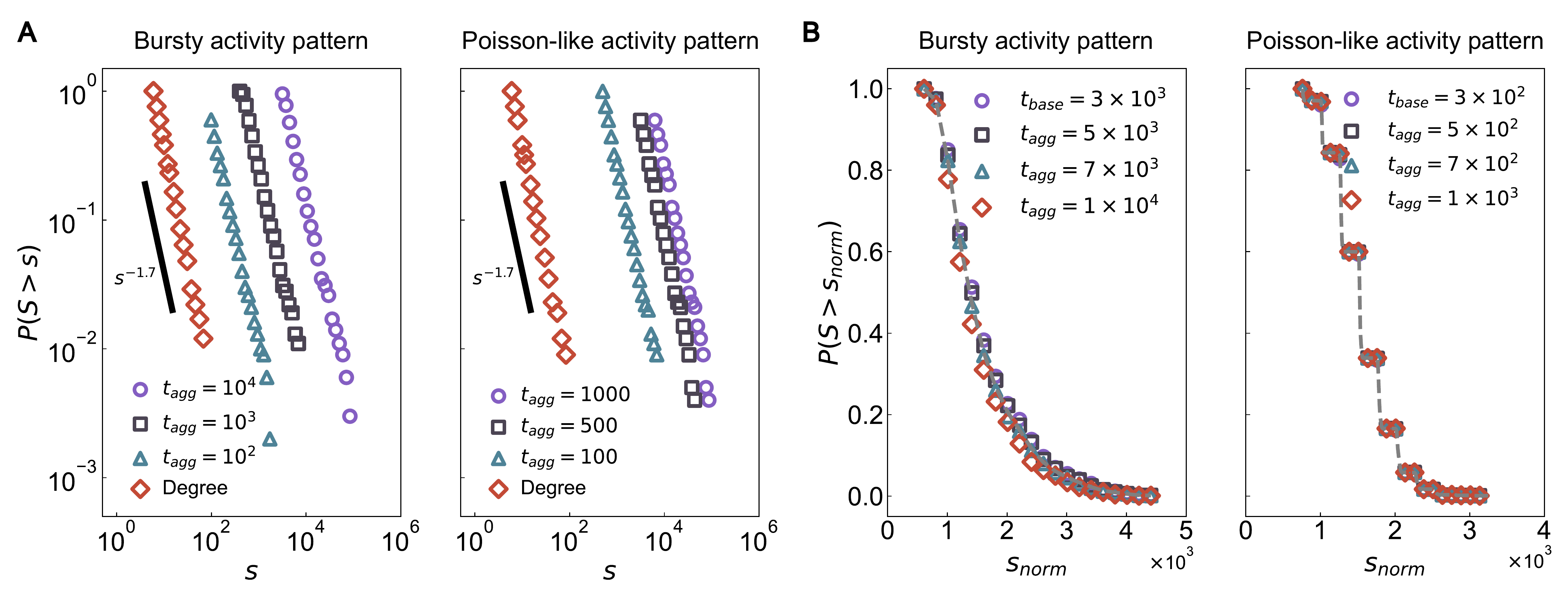}
\centering
\caption{\textbf{Robustness of node strength distributions on synthetic temporal networks.} 
We consider two classes of underlying topologies, Barabási-Albert scale-free networks (\textbf{A}) and Watts-Strogatz small-world networks (\textbf{B}).
For the first class, the survival function of the degree distribution is a power-law distribution with an exponent $\gamma \approx 1.7$ (red diamonds).
The survival function of the node strength of aggregated networks with different aggregation times $t_{agg}$ is all power-law distributions with an exponent $1.7$ in both activity patterns. 
For the second class, the normalized node strength distribution of aggregated networks is robust to the aggregation time $t_{agg}$ in both activity patterns. The targeted exponent of single nodes and links is $(\alpha_{pmf}, \beta_{pmf}) = (2.00, 1.90)$ in the bursty activity pattern and $(\alpha_{pmf}, \beta_{pmf}) = (2.50, 2.00)$ in the Poisson-like activity pattern.
}
\end{figure}

\end{document}